\documentclass[11pt]{article}
\usepackage[a4paper]{geometry} 
\usepackage{verbatim}
\usepackage{float}
\usepackage{amsthm}
\usepackage{amsmath}
\usepackage{amssymb}
\usepackage{graphicx}

\usepackage{tikz}
\usetikzlibrary{calc}

\makeatletter

\floatstyle{ruled}
\newfloat{algorithm}{tbp}{loa}
\providecommand{\algorithmname}{Algorithm}
\floatname{algorithm}{\protect\algorithmname}

\theoremstyle{plain}
\newtheorem{thm}{\protect\theoremname}
  \theoremstyle{definition}
  \newtheorem{problem}{\protect\problemname}
  \theoremstyle{definition}
  \newtheorem{defn}{\protect\definitionname}
  \theoremstyle{plain}
  \newtheorem{lem}{\protect\lemmaname}
  
	\newtheorem{cor}{Corollary}

\@ifundefined{date}{}{\date{}}
\usepackage{authblk}
\usepackage{algpseudocode}
\usepackage{algorithm}\usepackage{enumitem}

\newcommand{\expect}[1]{\mathbb{E}\left[{#1}\right]}

\usepackage{amsthm}
\usepackage{amsfonts}\usepackage{xspace}
\usepackage{upgreek}
\usepackage{braket}
\usepackage{enumitem}
\usepackage{amsfonts}

\usepackage[backend=bibtex,firstinits,style=numeric-comp,maxnames=99]{biblatex}% Numbered citations: style=numeric-comp or alphabetic

\DeclareFieldFormat{titlecase}{\MakeTitleCase{#1}}

\newrobustcmd{\MakeTitleCase}[1]{%
  \ifthenelse{\ifcurrentfield{booktitle}\OR\ifcurrentfield{booksubtitle}%
    \OR\ifcurrentfield{maintitle}\OR\ifcurrentfield{mainsubtitle}%
    \OR\ifcurrentfield{journaltitle}\OR\ifcurrentfield{journalsubtitle}%
    \OR\ifcurrentfield{issuetitle}\OR\ifcurrentfield{issuesubtitle}%
    \OR\ifentrytype{REMOVEIFEXCLUDEBOOKbook}\OR\ifentrytype{mvbook}\OR\ifentrytype{bookinbook}%
    \OR\ifentrytype{booklet}\OR\ifentrytype{suppbook}%
    \OR\ifentrytype{collection}\OR\ifentrytype{mvcollection}%
    \OR\ifentrytype{suppcollection}\OR\ifentrytype{manual}%
    \OR\ifentrytype{periodical}\OR\ifentrytype{suppperiodical}%
    \OR\ifentrytype{proceedings}\OR\ifentrytype{mvproceedings}%
    \OR\ifentrytype{reference}\OR\ifentrytype{mvreference}%
    \OR\ifentrytype{report}\OR\ifentrytype{thesis}}
    {#1}
    {\MakeSentenceCase{#1}}}
\bibliography{longnames,bib-latest,bristol-ESA}

\renewcommand{\geq}{\geqslant}
\renewcommand{\leq}{\leqslant}
\renewcommand{\epsilon}{\varepsilon}
\renewcommand{\Delta}{\Updelta}

\newcommand{\calA}{\ensuremath{\mathcal{A}}}

\newcommand{\calI}{\ensuremath{\mathcal{I}}}

\newcommand{\calT}{\ensuremath{\mathcal{T}}}
\newcommand{\calU}{\ensuremath{{\mathcal{U}}}}

\newcommand{\Hamarray}{\ensuremath{\textup{HamArray}}}

\newcommand{\Prob}{\ensuremath{\textup{Pr}}}
\newcommand{\expected}[1]{\ensuremath{\mathbb{E}[#1]}}

\newcommand{\ball}[1]{\ensuremath{\textup{Ball}(#1)}}
\newcommand{\Vsum}{\ensuremath{\textup{Sum}}}
\newcommand{\vsum}{\ensuremath{\textup{sum}}}

\newcommand{\lit}{large information transfer\xspace}
\newcommand{\MFL}{\ensuremath{M_{F,\ell}}}

\newcommand{\Xv}{\ensuremath{U_v}}
\newcommand{\Xvknown}{\ensuremath{\widetilde{U}_v}}
\newcommand{\Xvfix}{\ensuremath{\widetilde{u}_v}}
\newcommand{\Yv}{\ensuremath{A_v}}

\newcommand{\vitset}{\ensuremath{\mathcal{I}_v}}
\newcommand{\vit}{\ensuremath{I_v}}

\renewcommand{\S}{{\ensuremath{S}}} % Dynamic string (stream)
\newcommand{\U}{{\ensuremath{U}}} % Update string/sequence (text)
\newcommand{\arrive}{\ensuremath{\textsc{update}}}

\newcommand{\Cbig}{\ensuremath{\widetilde{C}}}

\newcommand{\cbig}{\ensuremath{\widetilde{c}}}

\newcommand{\ELL}{\ensuremath{\mu}}
\newcommand{\ps}{\ensuremath{\rho}}
\newcommand{\psymb}{\ensuremath{\star}}

\newcommand{\tsymb}{\ensuremath{\diamond}}

\newcommand{\alignment}[1]{{\small{\textcircled{\scriptsize #1}}}}

\newcommand{\insertdiagram}[1]{\includegraphics[scale=0.82]{figures/#1}}
\newcommand{\Patrascu}{P{\v a}tra{\c s}cu\xspace}

\newcommand{\inserttreefigure}{
\newcommand{\nodefont}{\scriptsize}
\begin{figure}[t]
\centering
\begin{tikzpicture}[level distance=3.9mm]
\tikzstyle{every node}=[fill=black!100,circle,inner sep=1pt]
\tikzstyle{level 1}=[sibling distance=72.5mm]
\tikzstyle{level 2}=[sibling distance=36mm]
\tikzstyle{level 3}=[sibling distance=18.2mm]
\tikzstyle{level 4}=[sibling distance=9.2mm]
\tikzstyle{level 5}=[sibling distance=4.5mm,set style={{every node}+=[fill=white!15]}]
\node {}
    child {node {}
        child {node {}
            child {node {}
                child {node {}
                    child {node {\nodefont 0}}
                    child {node {\nodefont 1}}
                }
                child {node {}
                    child {node {\nodefont 2}
                    }
                    child {node {\nodefont 3}
                    }
                }
            }
            child {node {}
                child {node {}
                    child {node {\nodefont 4}}
                    child {node {\nodefont 5}}
                }
                child {node {}
                    child {node {\nodefont 6}
                    }
                    child {node {\nodefont 7}
                    }
                }
            }
        }
        child {node {}
            child {node {}
                child {node {}
                    child {node {\nodefont 8}}
                    child {node {\nodefont 9}}
                }
                child {node {}
                    child {node {\nodefont 10}
                    }
                    child {node {\nodefont 11}
                    }
                }
            }
            child {node {}
                child {node {}
                    child {node {\nodefont 12}}
                    child {node {\nodefont 13}}
                }
                child {node {}
                    child {node {\nodefont 14}
                    }
                    child {node {\nodefont 15}
                    }
                }
            }
        }
    }
    child {node[fill=black!00,circle,draw,inner sep=1.5pt] {$v$}
        child {node {}
            child {node {}
                child {node {}
                    child {node {\nodefont 16}}
                    child {node {\nodefont 17}}
                }
                child {node {}
                    child {node {\nodefont 18}
                    }
                    child {node {\nodefont 19}
                    }
                }
            }
            child {node {}
                child {node {}
                    child {node {\nodefont 20}}
                    child {node {\nodefont 21}}
                }
                child {node {}
                    child {node {\nodefont 22}
                    }
                    child {node {\nodefont 23}
                    }
                }
            }
        }
        child {node {}
            child {node {}
                child {node {}
                    child {node {\nodefont 24}}
                    child {node {\nodefont 25}}
                }
                child {node {}
                    child {node {\nodefont 26}
                    }
                    child {node {\nodefont 27}
                    }
                }
            }
            child {node {}
                child {node {}
                    child {node {\nodefont 28}}
                    child {node {\nodefont 29}}
                }
                child {node {}
                    child {node {\nodefont 30}
                    }
                    child {node {\nodefont 31}
                    }
                }
            }
        }
    };
\end{tikzpicture}

\caption{\label{fig:tree}An information transfer tree $\calT$ with $n=32$ leaves. For the node labelled $v$, the arrival times $t_0=16$, $t_1=23$ and $t_2=31$.}
\end{figure}
}

\makeatother

  \providecommand{\definitionname}{Definition}
  \providecommand{\lemmaname}{Lemma}
  \providecommand{\problemname}{Problem}
\providecommand{\theoremname}{Theorem}

\begin{document}

%\pagenumbering{gobble}

\title{Time Bounds for Streaming Problems}

\author{Rapha\"{e}l Clifford \quad ~ Markus Jalsenius \quad ~ Benjamin Sach\\
 Department of Computer Science\\ University of Bristol\\ Bristol, UK}

\date{}

\maketitle

%\newpage
%\pagenumbering{arabic}
%\cleardoublepage
%\setcounter{page}{1}

%\pagenumbering{arabic}
%\setcounter{page}{1}%Leave this line commented out.

\begin{abstract}

We give tight cell-probe bounds for the time to compute convolution, multiplication and Hamming distance in a stream. The cell probe model is a particularly strong computational model and subsumes, for example, the popular word RAM model.
\begin{itemize}
 \item We first consider online convolution where the task is to output the inner product between a fixed $n$-dimensional vector and a vector of the $n$ most recent values from a stream. One symbol of the stream arrives at a time and the each output must be computed before the next symbols arrives.
 \item Next we show bounds for online multiplication where the stream consists of pairs of digits, one from each of two $n$ digit numbers that are to be multiplied. One pair arrives at a time and the task is to output a single new digit from the product before the next pair of digits arrives. 
 \item Finally we look at the online Hamming distance problem where the Hamming distance is outputted instead of the inner product. 
 
\end{itemize}

For each of these three problems, we give a lower bound of $\Omega\left(\frac{\delta}{w}\log n\right)$ time on average per output, where $\delta$ is the number of bits needed to represent an input symbol and $w$ is the cell or word size. We argue that these bound are in fact tight within the cell probe model. 
\end{abstract}

%%%%%%%%%%%%%%%%%%%%%%%%%%%%%%%%%%%%%%%%%%%%%%%%%%%%%%%%%%%%%%%%%%%%%%%%%%%%%%%%%%%%%%%%%
%%%%%%%%%%%%%%%%%%%%%%%%%%%%%%%%%%%%%%%%%%%%%%%%%%%%%%%%%%%%%%%%%%%%%%%%%%%%%%%%%%%%%%%%%
%%%%%%%%%%%%%%%%%%%%%%%%%%%%%%%%%%%%%%%%%%%%%%%%%%%%%%%%%%%%%%%%%%%%%%%%%%%%%%%%%%%%%%%%%
%%
%% Convolution
%%
%%%%%%%%%%%%%%%%%%%%%%%%%%%%%%%%%%%%%%%%%%%%%%%%%%%%%%%%%%%%%%%%%%%%%%%%%%%%%%%%%%%%%%%%%
%%%%%%%%%%%%%%%%%%%%%%%%%%%%%%%%%%%%%%%%%%%%%%%%%%%%%%%%%%%%%%%%%%%%%%%%%%%%%%%%%%%%%%%%%
%%%%%%%%%%%%%%%%%%%%%%%%%%%%%%%%%%%%%%%%%%%%%%%%%%%%%%%%%%%%%%%%%%%%%%%%%%%%%%%%%%%%%%%%%
%%%%%%%%%%%%%%%%%%%%%%%%%%%%%%%%%%%%%%%%%%%%%%%%%%%%%%%%%%%%%%%%%%%%%%%%%%%%%%%%%%%%%%%%%
%%%%%%%%%%%%%%%%%%%%%%%%%%%%%%%%%%%%%%%%%%%%%%%%%%%%%%%%%%%%%%%%%%%%%%%%%%%%%%%%%%%%%%%%%

\section{Introduction}

We consider the complexity of three related and fundamental problems: computing the convolution of two vectors, multiplying two integers, and computing the Hamming distance between two strings. We study these problems in an online or streaming context and provide matching upper and lower bounds in the cell-probe model.
Lower bounds in the cell-probe model also hold for the popular word-RAM model in which many of today's algorithms are given.

The importance of these problems is hard to overstate. The integer multiplication and convolution problems have played a central role in modern algorithms design and theory.  The question of how to compute the Hamming distance efficiently has a rich literature, spanning many of the most important fields in computer science. Within the theory community, communication complexity based lower bounds and streaming model upper bounds for the Hamming distance problem have been the subject of particularly intense study~\cite{CDIM:03,Woodruff:04,HSZZ:06,JKS:08,BCRT:10,Chakrabarti:11}. This previous work has however almost exclusively focussed on providing resource bounds either in terms of space or bits of communication rather than time complexity.

We begin by introducing the problems and stating our results. In the following problem definitions and throughout, we write $[q]$ to denote the set $\{0,\dots,q-1\}$, where $q$ is a positive integer and a parameter of the problem.

\begin{problem}[\textbf{Online convolution}]
    For a fixed vector $F\in[q]^n$ of length $n$, we consider a stream in which numbers from $[q]$ arrive one at a time. For each arriving number, before the next number arrives, we output the inner product (modulo~$q$) of $F$ and the vector that consists of the most recent  $n$ numbers of the stream.
\end{problem}

%For the online convolution problem we have the following result.

\begin{thm}[\textbf{Online convolution}]
    \label{thm:conv}
    In the cell-probe model with $w$ bits per cell, for any positive integers~$q$ and~$n$, and any randomised algorithm solving the online convolution problem, there exist instances such that the expected amortised time per arriving value is $\Omega{\left(\frac{\delta}{w}\log n\right)}$, where $\delta=\lceil \log_{2}{q} \rceil$.
\end{thm}

%We show that there are instances of this problem such that any algorithm solving it will require $\Omega(\frac{\delta}{w}\log n)$ amortised time per output, where $\delta=\log_2 q$ and $w$ is the number of bits per cell in the cell-probe model. The result is formally stated in Theorem~\ref{thm:conv}.

%\begin{problem}[\textbf{Online multiplication}]
%    \margin{change so that one operand is fixed, and update text later on.}
%    Given two numbers $X,Y\in [q^n]$, where $q$ is the base and $n$ is the number of digits per number, we want to output the $n$ least significant digits of the product of $X$ and $Y$, in base $q$. We must do this under the constraint that the $i$-th digit of the product (starting from the lower-order end) is outputted before the $(i+1)$-th digit, and when the $i$-th digit is outputted, we only have access to the $i$ least significant digits of $X$ and $Y$, respectively. We can think of the digits of $X$ and $Y$ arriving online in pairs, one digit from each of $X$ and $Y$.
%\end{problem}

\begin{problem}[\textbf{Online multiplication}]
    Given two numbers $F,X\in [q^n]$, where $q$ is the base and $n$ is the number of digits per number, we want to output the $n$ least significant digits of the product of $F$ and $X$, in base $q$.
    We must do this under the constraint that only $F$ is known in advance and the digits of $X$ arrive one at a time, starting from the lower-order end.
    When the $i$-th digit of $X$ arrive, before the $(i+1)$-th digit arrive, we output the $i$-th digit of the product.
\end{problem}

\begin{thm}[\textbf{Online multiplication}]
    \label{thm:mult}
    In the cell-probe model with $w$ bits per cell, for any positive integers~$q$ and~$n$, and any randomised algorithm solving the online multiplication problem in base~$q$, there exist instances such that outputting the $n$ least significant digits of the product takes $\Omega{\left(\frac{\delta}{w}n\log n\right)}$ expected time, where $\delta=\lceil \log_2q \rceil$.
\end{thm}

%We prove a slightly stronger version of Theorem~\ref{thm:mult}, namely, we assume that one of the numbers to be multiplied is not part of the input, hence is known to the algorithm. Only the digits of the other number arrive one at a time.

%We show that there are instances of this problem such that any algorithm solving it takes $\Omega(\frac{\delta}{w} \log n)$ time on average per input pair, where $\delta=\log_2 q$ and $w$ is the number of bits per cell in the cell-probe model. In fact, our lower bound will even hold under the restriction that one of the two numbers is known in advance it's entirety. The result is formally stated in Theorem~\ref{thm:mult}.

\begin{problem}[\textbf{Online Hamming distance}]
    For a fixed string $F$ of length $n$, we consider a stream in which symbols from the alphabet $[q]$ arrive one at a time. For each arriving symbol, before the next symbol arrives, we output the Hamming distance between $F$ and the last $n$ symbols of the stream.
\end{problem}

\begin{thm}[\textbf{Online Hamming distance}]
    \label{thm:ham}
    In the cell-probe model with $w$ bits per cell, for any positive integers~$q$ and~$n$, and any randomised algorithm solving the online Hamming distance problem, there exist instances such that the expected amortised time per arriving value is $\Omega\left(\frac{\delta}{w}\log n\right)$, where $\delta=\lceil \min\{\log_2q,\log_2n\} \rceil$.
\end{thm}

% Unlike the convolution and multiplication problems, a lower bound for the Hamming distance problem cannot scale with an arbitrary alphabet size as there can be at most $n$ distinct symbols in the fixed string $F$.

%We show that there are instances of this problem for which any algorithm solving it will require $\Omega\left(\frac{\delta}{w}\log n\right)$ time on average per output, where   $\delta$ is the number of bits needed to represent an input symbol and $w$ is the number of bits per cell in the cell-probe model. Lower bounds in the cell-probe model also hold for the popular word-RAM model in which many of today's algorithms are given.  The full statement is given in Theorem~\ref{thm:lower}.

Our Hamming distance lower bound also implies a matching lower bound for any problem that Hamming distance can be reduced to.  The most straightforward of these is online $L_1$ distance computation, where the task is to output the $L_1$ distance between a fixed vector of integers and the last $n$ numbers in the stream.
A suitable reduction was shown in~\cite{LP:2008}. The expected amortised cell probe complexity for the online $L_1$ distance problem is therefore also $\Omega{\left(\frac{\delta}{w}\log n\right)}$ per new output.

One of our main technical innovations is to extend recently developed methods designed to give lower bounds on dynamic data structures to the seemingly distinct field of online algorithms.  Where $\delta = w$, for example, we have $\Omega(\log{n})$ lower bounds for all three problems. In particular for online multiplication and convolution, these lower bounds match the currently best known offline upper bounds in the RAM model. As we discuss in the Section~\ref{sec:previous}, this may be the highest lower bound that can be formally proved for all the problems we consider without a further significant theoretical breakthrough.

In order to prove our lower bounds we show the existence of probability distributions on the inputs for which we can prove lower bounds on the expected running time of any deterministic algorithm. By Yao's minimax principle~\cite{Yao1977:Minimax} this immediately implies that for every (randomised) algorithm there is a worst-case input such that the (expected) running time is equally high. Therefore our lower bounds hold equally for randomised algorithms as for deterministic ones.

The lower bounds we give are also tight within the cell-probe model. This can be seen by application of reductions described in~\cite{FS:1973, CEPP:2011}.  It was shown there that any offline algorithm for convolution~\cite{CEPP:2011} or multiplication~\cite{FS:1973} can be converted to an online one with at most an $O(\log{n})$ factor overhead. For details of these reductions we refer the reader to the original papers.  In our case, the same approach also allows us to directly convert any cell-probe algorithm from an offline to online setting. An offline cell-probe algorithm for convolution, multiplication or Hamming distance could first read the whole input, then compute the answers and finally output them. This takes $O{(\frac{\delta}{w} n)}$ cell probes. We can therefore derive online cell-probe algorithms which take only $O{(\frac{\delta}{w}n\log n)}$ probes over $n$ inputs, hence $O{(\frac{\delta}{w}\log n)}$ (amortised) probes per output. This upper bound matches the new lower bounds we give.
We summarise this in the following corollary.
%As an offline cell-probe algorithm for either multiplication or convolution requires  $O{(\frac{\delta}{w} n)}$ probes, we can derive online cell-probe algorithms which take only $O{(\frac{\delta}{w}\log n)}$ probes per new (pair of) symbols which therefore match the new lower bounds we give.

\begin{cor}\label{cor:final}
 The expected amortised cell-probe complexity of the online convolution, multiplication, Hamming distance and $L_1$-distance problems is $\Theta(\frac{\delta}{w}\log n)$ per arriving value.
\end{cor}

One consequence of our results is the first strict separation between the complexity of exact and inexact pattern matching.  Online exact matching can be solved in constant time~\cite{Galil:1981} per new input symbol and our new lower bound proves for the first time that this is not possible for Hamming distance.

%Further more we have shown a similar separation for algorithms which make use of fast convolution computation, which is a common key ingredient in many of today's most efficient algorithms in pattern matching.

Another consequence of our results is a new separation between the time complexity of online exact matching and any convolution-based online pattern matching algorithm. Convolution has played a particularly important role in the field of combinatorial pattern matching where many of the fastest algorithms rely crucially for their speed on the use of fast Fourier transforms (FFTs) to perform repeated convolutions. These methods have also been extended to allow searching for patterns in rapidly processed data streams~\cite{CEPP:2011,CS:2011}.

\subsection{Previous results and upper bounds in the RAM model} \label{sec:previous}

Almost all previous algorithmic work for exact Hamming distance computation has considered the problem in an offline setting.  Given a pattern~$P$ and a text~$T$ of length $m$ and $n$ respectively, the best current deterministic upper bound for offline Hamming distance computation is an $O(n\sqrt{m\log{|m|}})$ time algorithm based on convolutions~\cite{Abrahamson:1987, Kosaraju:1987}. In \cite{Karloff:1993} a randomised algorithm was given that takes $O((n/{\epsilon}^2)\log^2{n})$ time which was subsequently modified in~\cite{Indyk:1998} to $O((n/{{\epsilon}^3}) \log{n})$. Particular interest has also been paid to a bounded version of this problem called the $k$-mismatch problem.  Here a bound $k$ is given and we need only report the Hamming distance if it is less than or equal to $k$. In \cite{LV:1986a}, an $O(nk)$ algorithm was given that is not convolution based and uses $O(1)$ time lowest common ancestor (LCA) operations on the suffix tree of $P$ and $T$. This was then  improved to $O(n\sqrt{k\log{k}})$ time by a method that combines LCA queries, filtering and convolutions~\cite{ALP:2004}.

The best time complexity lower bounds for online multiplication of two $n$-bit numbers were given in the 1974 by Paterson, Fischer and Meyer.  They presented an $\Omega(\log{n})$ lower bound for multitape Turing machines~\cite{PFM:1974} and also gave an $\Omega(\log{n}/\log{\log n})$ lower bound for the \emph{bounded activity machine} (BAM).  The BAM, which is a strict generalisation of the Turing machine model but which has nonetheless largely fallen out of favour, attempts to capture the idea that future states can only depend on a limited part of the current configuration.  To the authors' knowledge, there has been no progress on cell-probe lower bounds for online multiplication, convolution or Hamming distance previous to the work we present here.

There have however been attempts to provide offline lower bounds for the related problem of computing the FFT.  In~\cite{Morgenstern:1973} Morgenstern gave an $\Omega(n \log{n})$ lower bound conditional on the assumption that the underlying field of the transform is the complex numbers and that the modulus of any complex numbers involved in the computation is at most one.  Papadimitriou gave the same $\Omega(n \log{n})$ lower bound for FFTs of length a power of two, this time excluding certain classes of algorithms including those that rely on linear mathematical relations among the roots of unity~\cite{Papadimitriou:1979}.  This work had the advantage of giving a conditional lower bound for FFTs over more general algebras than was previously possible, including for example finite fields.  In 1986, Pan~\cite{Pan:1986} showed that another class of algorithms having a so-called synchronous structure must require $\Omega(n \log{n})$ time for the computation of both the FFT and convolution.

The fastest known algorithms for both offline integer multiplication and convolution in the word-RAM model require $O(n\log{n})$ time by a well known application of a constant number of FFTs.  As a consequence our online lower bounds for these two problems match the best known time upper bounds for the offline problem. As we discussed above, our lower bounds for all three problems are also tight within the cell-probe model for the online problems.   

The question now naturally arises as to whether one can find higher lower bounds in the RAM model.   This appears as an interesting question as there remains a gap between the best known time upper bounds provided by existing algorithms and the lower bounds that we give within the cell-probe model.  However, as we mention above, any offline algorithm for convolution, Hamming distance or multiplication can be converted to an online one with at most an $O(\log{n})$ factor overhead~\cite{FS:1973,CEPP:2011}.
As a consequence, a higher lower bound than $\Omega(\log{n})$ for any of these problems would immediately imply a superlinear lower bound for the offline version of the corresponding problem. This would be a truly remarkable breakthrough in the field of computational complexity as no such offline lower bound is known even for the canonical NP-complete problem SAT.

Our only alternative route to find tight time bounds would be to find better upper bounds for the online problems. For the case of online multiplication at least, where the fastest online RAM algorithm takes $O(\log^2{n})$ time per arriving pair of digits, this has been an open problem since at least 1973 and has so far resisted our best attempts.
On the other hand, for online Hamming distance, while our lower bound is tight within the model, it is still distant from the time complexity of the fastest known RAM algorithms. The best known online complexity is $O(\sqrt{n\log{n}})$ time per arriving symbol~\cite{CEPP:2011}. An improvement of the upper bound for Hamming  distance computation to meet our new lower bound would also have significant implications. A reduction that is now regarded as folklore tells us that any $O(f(n))$~time algorithm for computing the Hamming distance between a pattern and all substrings of a text, assuming a pattern of length~$n$ and a text of length~$2n$, implies an $O(f(n^2))$~time algorithm for multiplying binary $(n\!\times\! n)$-matrices over the integers. Therefore an $O(\log{n})$ time online Hamming distance algorithm would imply an $O(n\log{n})$  offline Hamming distance algorithm, which would in turn imply an $O(n^2\log{n})$ time algorithm for binary matrix multiplication.  Although such a result would arguably be less shocking than a proof of a superlinear offline lower bound for Hamming distance computation, it would nonetheless be a significant breakthrough in the complexity of a classic and much studied problem.

\subsection{The cell-probe model}

Our bounds hold in the \emph{cell-probe model} which is a particularly strong computational model that was introduced originally by Minsky and Papert~\cite{MP:1969} in a different context and then subsequently by Fredman~\cite{Fredman:1978} and Yao~\cite{Yao1981:Tables}.
In the cell-probe model there is a separation between the computing unit and the memory, which is external and consists of a set of cells of $w$ bits each. The computing unit cannot remember any information between operations. Computation is free and the cost is measured only in the number of cell reads or writes (cell~probes). This general view makes the model very strong, subsuming for instance the popular word-RAM model. In the word-RAM model certain operations on words, such as addition, subtraction and possibly multiplication take constant time (see for example~\cite{Hagerup:1998} for a detailed introduction). Here a word corresponds to a cell. As is typical, we will require that the cell size $w$ is at least of order $\log n$ bits. This allows each cell, or a constant number of cells, to hold the address of any location in memory.

The generality of the cell-probe model makes it particularly attractive for establishing lower bounds for dynamic data structure problems and many such results have been given in the past couple of decades. The approaches taken had historically been based only on communication complexity arguments and the chronogram technique of Fredman and Saks~\cite{FS1989:chronogram}.
%There remains however, a number of unsatisfying gaps between the lower bounds and known upper bounds.
However in 2004, a breakthrough lead by \Patrascu and Demaine gave us the tools to seal the gaps for several data structure problems~\cite{PD2006:Low-Bounds} as well as giving the first $\Omega(\log{n})$ lower bounds. The new technique is based on information theoretic arguments that we also deploy here. \Patrascu and Demaine also presented ideas which allowed them to express more refined lower bounds such as trade-offs between updates and queries of dynamic data structures.  For a list of data structure problems and their lower bounds using these and related techniques, see for example~\cite{Pat2008:Thesis}.  More recently, a new lower bound of $\Omega\left((\log{n}/\log{\log{n}})^2\right)$ was given by Green Larsen for the cell-probe complexity of performing queries in the dynamic range counting problem~\cite{Larsen:2012}.
This result  holds under the natural assumptions of $\Theta(\log{n})$~size words and polylogarithmic time updates and is another exciting breakthrough in the field of cell-probe complexity.

%\paragraph{Our contributions}
\subsection{Technical contributions}

We use one of the most important techniques for proving data structure lower bounds called the \emph{information transfer method} of \Patrascu and Demaine~\cite{PD2004:Partial-sums,PD2006:Low-Bounds}.
For a pair of adjacent intervals of arriving values in the stream, the information transfer is the set
of memory cells that are written during the first interval and read in
the next interval. These cells must contain \emph{all} the
information from
the updates during the first interval that the algorithm needs in order 
to produce correct outputs in the next interval. If
one can prove that this quantity is large for many pairs of intervals
then the desired lower bounds follow. To do this we relate the size
of the information transfer to the conditional entropy of the outputs
in the relevant interval. The main task of proving lower bounds
reduces to that of devising a hard input distribution for which outputs
have high entropy conditioned on selected previous values of the input.

Although the use of information transfer to provide time lower bounds for data structure problems is not new,  applying the method to our new online setting has required a number of new insights and technical innovations. At the simplest level, where a standard data structure problem has a number of different possible queries, in our setting there is only one query which is to return the latest result as soon as a  new symbol arrives. As a result we provide a complete description of the information transfer method in a form which is relevant to this different setting.  At a more detailed mathematical level, perhaps the most surprising innovation we present is a new relationship between the Hamming distance, vector sums and constant weight binary cyclic codes.

For the three problems we consider, our key innovation is the design of a fixed vector or string $F$ 
which together with some random distribution over possible input streams provide a lower bound for the information transfer between successive intervals.
For the convolution and multiplication problems we show that a randomly picked $F$ has a good chance of being suitable for proving the lower bounds. We also give an explicit description of a particular $F$ for which the lower bounds are obtained when the values of the input stream are drawn independently and uniformly at random. The vector $F$ is easy to describe and naturally yields large conditional entropy of the outputs for intervals of power-of-two lengths.

The results of the convolution and multiplication problems can be seen as a first step towards the lower bound for the Hamming distance problem.
Here the string $F$ is derived by a sequence of transformations. These start with binary cyclic codes and go via binary vectors with many 
distinct sums and an intermediate string to finally arrive at $F$ itself.
The use of such a purposefully designed input departs from the closely related work of the convolution and multiplication lower bounds and also from much of the lower bound literature where simple uniform distributions over the whole input space often suffice.
%It also however necessarily creates a number of challenging technical hurdles which we overcome.

The central fact that enabled a lower bound to be proven for the online convolution problem is that the inner product  between a vector and successive suffixes of the stream reveals a lot of information about the history of the stream.  Establishing a similar result for online Hamming distance problem appears, however, to be considerably more challenging for a number of reasons.
The first and most obvious is that the amount of information one gains by comparing whether two, potentially large, symbols are equal is at most one bit, as opposed to $O(\log{n})$ bits for multiplication.
The second is that the particularly simple worst-case vector $F$ of the convolution problem greatly eased the resulting analysis.
We have not been able to find such a simple fixed string for the Hamming distance problem and our proof of the existence of a hard instance is non-constructive and involves a number of new insights, combining ideas from coding theory and additive combinatorics.

When computing the Hamming distance there is a balance between the number of symbols being used and the length of the strings. For large alphabets and short strings, one would expect a typical outputted Hamming distance to be close to the length of the string on random inputs and therefore to provide very little information.
This suggests that the length of the strings must be sufficiently long in relation to the alphabet size to ensure that the entropy of the outputs is large, as required by the information transfer method. On a closer look, it is not immediately obvious that large entropy can be obtained unless the fixed string $F$ is \emph{exponentially} larger than the alphabet size. This potentially poses another problem for the information transfer method, namely that a word size $w$ of order $\log n$ would be much larger than $\delta$ (the number of bits needed to represent a symbol), making a $\log n$ lower bound impossible to achieve.
%Further, it is not even clear how to construct a fixed string such that for windows of exponential length in the alphabet size, the entropy of the outputted Hamming distances is proportional to the entropy of the symbols within the window.
%entropy could end up being \emph{too small} in relation to the number of inputs over which the method operates. Namely, over short windows of the text stream, the Hamming distance outputs provide very little information

Our main technical contribution is to show that fixed strings of length only polynomial in the size of the alphabet exist which provide outputs of sufficiently high entropy.  Such strings, when combined with a suitable input distribution maximising the number of distinct Hamming distance output sequences, give us the overall lower bound.
We design a fixed string $F$ with this desirable property in such a way that there is a one-to-one mapping between many of the different possible input streams and the outputted Hamming distances. This in turn implies large entropy.
The construction of $F$ is non-trivial and we break it into smaller building blocks, reducing our problem to a purely combinatorial question relating to vectors sums.
That is, given a relatively small set $V$ of vectors of length $m$, how many distinct vector sums can be obtained by choosing $m$ vectors from $V$ and adding them. We show that even if we are restricted to picking vectors only from subsets of $V$, there exists a $V$ such that the number of distinct vector sums is  $m^{\Omega(m)}$. We believe this result is interesting in its own right. Our proof for the combinatorial problem is non-constructive and probabilistic, using constant weight cyclic binary codes to prove that there is a positive probability of the existence of a set $V$ with the desired property.

%Our main technical contribution which is the heart of our lower bound proof is to show that such strings do indeed exist with length $n' \in \Theta(2^{3\delta})$. Each of these  strings will have the property that given the right input distribution, it will be possible to reconstruct a constant fraction of the stream from the previous time interval using successive outputs to the online Hamming distance problem.   In this way, the amount of information transferred from one time interval to the next will be high. Unfortunately  we cannot simply use such strings directly to get our cell-probe lower bound.  We can and do however use them as building blocks.  The worst case string $F$ is built up from these high entropy producing shorter strings, eventually giving a string which provides the final lower bound.

\subsection{Organisation}

In Section~\ref{sec:preliminaries} we introduce notation and describe the setup for proving the lower bounds.
In Section~\ref{sec:proofs} we prove the lower bounds for all three problems that we consider. The proofs hinge on a set of lemmas that will be proved separately in subsequent sections.
In Section~\ref{sec:conv} we deal with the lemmas related to the convolution problem, and
in Section~\ref{sec:mult} we deal with the lemmas related to the multiplication problem.
Finally, in Sections~\ref{sec:ham} to~\ref{sec:proofvecsum} we prove the lemma related to the Hamming distance problem.

\section{Basic setup for the lower bounds}\label{sec:preliminaries}

In this section we introduce notation and concepts that are used heavily
in the lower bound proofs.
%For a positive integer $n$, $[n]$ denotes the set $\{0,\dots,n-1\}$.
For an array, vector or string $A$ of length $n$ and $i,j\in[n]$, we write $A[i]$
to denote the value at position~$i$, and where $j\geq i$, $A[i,j]$
denotes the $(j-i+1)$-length subarray of $A$ starting at position
$i$. All logarithms are in base two.
We first introduce a unifying framework for the problems we consider.
%We differ discussion of integer multiplication within this framework until later and focus first on convolution and Hamming distance.

\subsection{The framework}
There is a \emph{fixed~array}
$F$ and an array $S$ which is referred
to as the \emph{stream}.
Both $F$ and $S$ are of length $n$ and over the set $[q]$ of integers,
%, referred to as the \emph{alphabet},
and we let $\delta=\lfloor \log q \rfloor$ denote the number of bits required to encode a value from $[q]$.
The value $q$, or alternatively $\delta$, is a parameter of the problem.
%An element of
%the alphabet is often referred to as a \emph{symbol}.
The problem is to maintain $S$ subject to an update operation $\arrive(x)$ which takes a symbol $x\in [q]$, modifies~$S$ by appending $x$ to the right of the rightmost symbol $\S[n-1]$ and removing the leftmost symbol $S[0]$, and then outputs the value of a function of $F$ and the updated $S$.
In the \emph{convolution} problem the output is the inner product of $F$ and $S$, that is $\sum_{i\in[n]}(F[i]\cdot S[i])$, and in
the \emph{Hamming~distance} problem the output is the number of positions $i\in[n]$ such that $F[i]\neq S[i]$.

We let $U\in [q]^n$ denote the \emph{update array} which describes a sequence of  $n$ $\arrive$ operations. That is, for each $t\in[n]$, the operation $\arrive(U[t])$ is performed.
We will usually refer to $t$ as the \emph{arrival} of the value $U[t]$.
Observe that just after the arrival $t$, the values $U[t+1,n-1]$ are still not known to the algorithm. Finally, we let the $n$-length array $A$ denote the outputs such that for $t\in [n]$, $A[t]$ is the output of $\arrive(U[t])$.

In the \emph{multiplication}~problem we let $F$ denote one of the two operands to be multiplied, hence $F$ is fixed and known in advance by the algorithm.
Specifically we let $F[i]$ denote the $i$-th least significant digit. We let $U$ be the unknown operand so that $U[t]$ is its $t$-th least significant digit. Prior to the arrival of the first digit $U[0]$, the stream $S$ contains only zeros. The output $A[t]$ is the $t$-th digit in the product of $F$ and $S$, which is a function of $F$ and $U[0,t]$ as required.

\subsection{Hard distributions}

Our lower bounds hold for any randomised algorithm on its worst case
input. This will be achieved by applying \emph{Yao's
minimax principle}~\cite{Yao1977:Minimax}. That is, we develop
lower bounds that
hold for any deterministic~algorithm on some random~input. The basic approach is as follows: we devise a fixed array $F$ and describe a probability
distribution for $n$ new values arriving in the stream~$S$.
We then obtain a lower bound on the expected
running time for any deterministic algorithm over these arrivals.
Due to the minimax principle, the
same lower bound must then hold for any randomised~algorithm on its own
worst~case input. The amortised bound is obtained by dividing by $n$.

From this point onwards we consider an arbitrary deterministic algorithm
running with some fixed array $F$ on a random input of $n$ values. The algorithm may depend on~$F$. We refer to the choice of $F$ and distribution on $U$ as a \emph{hard distribution} since it used to show a lower bound.

%To obtain a randomised, \emph{worst-case} lower bound, it suffices to show the existence of such a hard distribution - in particular we give a single $F$ for each $n=|F|$. It is therefore natural to ask the question ``For which choices of $F$ does the lower bound hold?''. For the convolution problem and the integer multiplication problem, we answer this question by proving that our lower bounds hold even when $F$ is chosen uniformly at random. This is a slightly subtler claim and warrants clarification. Pick $F$ uniformly at random from $[q]^n$. Consider any Las Vegas randomised algorithm for convolution. This algorithm may depend on $F$.
%We then give a lower bound on the expected worst-case running time of this algorithm. The expectation is over the choice of $F$ while the worst-case input is over all sequences of $n$ stream updates. Therefore for most $F$, the lower bound holds.

\subsection{Information transfer} \label{sec:more-notation}

The \emph{information transfer tree}, denoted $\calT$, is a balanced
binary tree over $n$ leaves. To avoid technicalities we assume that
$n$ is a power of two.
The leaves
of $\calT$, from left to right, represent the arrivals $t$ from $0$
to $n-1$.
For a node $v$ of $\calT$, we let $\ell_{v}$
denote the number of leaves in the subtree rooted at $v$. 
An internal node $v$ is associated with three arrivals,
$t_{0}$, $t_{1}$ and $t_{2}$. Here $t_{0}$ is the arrival represented
by the leftmost node in subtree rooted at $v$, similarly $t_{2}=t_{0}+\ell_{v}-1$ is the rightmost such node and $t_{1}=t_{0}+\ell_{v}/2-1$ is in the middle. That is, the intervals $[t_{0},t_{1}]$
and $[t_{1}+1,t_{2}]$ span the left and right subtrees of $v$, respectively.
For example, in Figure~\ref{fig:tree},\inserttreefigure the node labelled $v$ is associated with the intervals $[16,23]$ and $[24,31]$.

We define the subarray $\Xv=U[t_0,t_1]$ to represent the $\ell_{v}/2$ values
arriving in the stream during the arrival interval $[t_{0},t_{1}]$, and we define the subarray $\Yv=A[t_{1}+1,t_{2}]$ to represent the $\ell_{v}/2$
outputs during the arrival interval $[t_{1}+1,t_{2}]$.

We define $\Xvknown$ to be the concatenation of $U[0,(t_0-1)]$
and $U[(t_2+1),(n-1)]$. That is, $\Xvknown$ contains all symbols
of $U$ except for those in $\Xv$.
When $\Xvknown$ is fixed to some constant $\Xvfix$ and $\Xv$ is
random, we write $H(\Yv\mid\Xvknown=\Xvfix)$ to denote the conditional
entropy of $\Yv$ under the fixed~$\Xvknown$.

We define the \emph{information transfer} of a node $v$ of $\calT$,
denoted $\vitset$, to be the set of memory cells $c$ such that $c$
is probed during the interval $[t_{0},t_{1}]$ and also probed in $[t_{1}+1,t_{2}]$.
The cells in
the information transfer $\vitset$ therefore contains
all the information about the values in $\Xv$ that the algorithm
uses in order to correctly produce the outputs $\Yv$.

By adding up the sizes of the information transfers $\vitset$ over the internal nodes $v$
of $\calT$ we get a lower bound on the number of cell probes, that is a lower bound on the total running time of the algorithm.
To see this it is important to make the observation
that a particular cell probe is counted for only once.
Suppose that the cell $c\in\vitset$ for some node $v$. Let $p$ be the first probe of $c$ in the arrival interval $[t_{1}+1,t_{2}]$. By including the cell $c\in\vitset$ in the cell probe count we are in fact counting the probe $p$. Now observe that $p$ cannot be counted for in the information transfer $\calI_{v'}$ of any node $v'$ where $v'$ is a proper descendant or ascendant of $v$.

Since the concept of the size of the information transfer is central to the lower bound proofs, we define as a shorthand $\vit=|\vitset|$ to denote the size of the information transfer.

\begin{defn}
[\textbf{Large expected information transfer}]\label{def:large-it}
A node $v$ of $\calT$ has \emph{\lit} if
\[
    \expected{\vit} ~\geq~ \frac{k\cdot\delta\cdot\ell_v}{w},
\]
where $k$ is a constant that depends on the problem and input distribution.
\end{defn}

The aim is to show that a substantial proportion of nodes of $\calT$ have \lit.

\section{Overall proofs of the lower bounds}\label{sec:proofs}

In this section we give the overall proofs for our lower bound
results.
Let $v$ be any node of $\calT$. Suppose that $\Xvknown$ is fixed but the
symbols in $\Xv$ are randomly drawn in accordance with the distribution
on $U$, conditioned on the fixed value of $\Xvknown$. This induces
a distribution on the outputs $\Yv$. If the entropy of $\Yv$ is
large, conditioned on the fixed $\Xvknown$, then any algorithm must
probe many cells in order to correctly produce the outputs $\Yv$, as it is
only through the information transfer $\vitset$ that the algorithm
can know anything about $\Xv$. We will soon make this claim more
precise.

\subsection{Upper bound on the entropy}

Towards showing that high conditional entropy  $H(\Yv\mid\Xvknown=\Xvfix)$ implies large information transfer we use the information transfer $\vitset$ to describe an encoding of the outputs $\Yv$. The following lemma gives a direct relationship between the size of the information transfer $\vitset$ and the entropy.
The lemma was originally stated in~\cite{PD2006:Low-Bounds} but for completeness we restate it here in our notation and provide a full proof.

\begin{lem}[\Patrascu and Demaine~\cite{PD2006:Low-Bounds}]\label{lem:H-upper-old}
Under the assumption that the address of any cell can be specified in $w$ bits, for any node $v$ of the information transfer tree $\calT$, the entropy 
$$
H(\Yv\mid\Xvknown=\Xvfix)~\leq~w + 2w\cdot \expected{\vit \mid \Xvknown=\Xvfix}.
$$
\end{lem}
\begin{proof}
The expected length of any encoding of $\Yv$, conditioned on $\Xvknown$, is an upper bound on the conditional entropy of $\Yv$.
We use the information transfer $\vitset$ as an encoding in the following way. For every cell $c\in\vitset$ we store the address of $c$, which takes at most $w$ bits under the assumption that a cell can hold the address of any cell in memory.
We also store the contents of $c$, which takes $w$ bits.
In total this requires $2w\cdot \vit$ bits.
We will use the algorithm, which is fixed, and the fixed values of $\Xvknown$ as part of the decoder to obtain $\Yv$ from the encoding. Since the encoding is of variable length we also store the size of the information transfer, which requires at most $w$ additional bits.

In order to prove that the described encoding of $\Yv$ is valid we now describe how to decode it.
First we simulate the algorithm on the fixed input $\Xvknown$ from the first arrival of $U[0]$ until just before the first value in $\Xv$ arrives.
We then skip over all inputs in $\Xv$
and resume simulating the algorithm from the beginning of the interval
where $\Yv$ is outputted until the last value in $\Yv$ has been obtained.
For every cell being read, we check if it is contained in information transfer $\vitset$ by looking up its address in the encoding.
If it is in the information transfer, its contents is fetched from the encoding. If not, its contents is available from simulating the algorithm on the fixed inputs.
Observe that it suffices to
store only the first time a cell in the information transfer is probed as the decoder
remembers every cell it has already accessed.
\end{proof}

\subsection{Lower bounds on the entropy}

Lemma~\ref{lem:H-upper-old} above provides a direct way to obtain a lower bound on the expected size of the information transfer if given a lower bound on the conditional entropy $H(\Yv\mid\Xvknown=\Xvfix)$.
To show that a node has \lit we introduce the following definition.

\begin{defn}
[\textbf{High-entropy node}]\label{def:high-node}A node $v$ in
$\calT$ is a \emph{high-entropy~node} if there is a positive constant
$k$ such that for \emph{any} fixed $\Xvfix$, 
$$
H(\Yv\mid\Xvknown=\Xvfix)\,\geq\, k\cdot\delta\cdot\ell_{v}.
$$
\end{defn}

To put this bound in perspective, note that the maximum conditional
entropy of $\Yv$ is bounded by the entropy of $\Xv$, which is at
most $\delta\cdot(\ell_{v}/2)$ and obtained when the values of $\Xv$
are independent and uniformly drawn from $[q]$. Thus, the
conditional entropy associated with a high-entropy node is the highest
possible up to some constant factor.
Establishing high-entropy nodes is the main contribution of this paper and the results are given in the following lemmas.

\begin{lem}
\label{lem:conv-rand-H-lower}
For the convolution problem,
suppose that $U$ is chosen uniformly at random from $[q]^n$, where $q$ is a prime. For any $v \in \calT$, at least a $(1-\frac{1}{q})$-fraction of all $F \in [q]^n$ have the property that $v$ is a high-entropy node.
\end{lem}

The proof of the above lemma is given in Section~\ref{sec:conv} and relies on properties of Toeplitz matrices over a finite field of $q$ elements. The proof does not give explicit descriptions of fixed arrays $F$ for which nodes are high-entropy nodes. 
In the proof of the next lemma however, we show that there exists a particular array $F$ for which high-entropy nodes are obtained. This $F$ is a 0/1-array and is easy to describe: zeroes everywhere except for at power-of-two positions from the right hand end. The proof is given in Section~\ref{sec:conv}.

\begin{lem}
\label{lem:conv-H-lower}
For the convolution problem there exists a fixed array $F\in [q]^n$ such that when $U$ is chosen uniformly at random from $[q]^n$, all $v \in \calT$ are high-entropy nodes.
\end{lem}

Before we give the lemmas concerning online multiplication, recall that in this problem there is a fixed operand $F$ multiplied with an operand $U$ for which digits arrive one at a time.

\begin{lem}
\label{lem:mult-rand-H-lower}
For the online multiplication problem, suppose that the operand $U$ is chosen uniformly at random from $[q^n]$. For any $v \in \calT$, at least half of all operands $F\in [q^n]$ have the property that $v$ is a high-entropy node.
\end{lem}

The proof of Lemma~\ref{lem:mult-rand-H-lower} is given in Section~\ref{sec:mult}. Similarly to the convolution problem we also give an explicit description of a number $F$ for which high-entropy nodes are obtained. This number resembles the fixed array that we described above for the convolution problem.
The proof of the next lemma is also given in Section~\ref{sec:mult}.

\begin{lem}
\label{lem:mult-H-lower}
For the online multiplication problem there exists a fixed operand $F\in [q^n]$ such that when $U$ is chosen uniformly at random from $[q^n]$, all $v \in \calT$ are high-entropy nodes.
\end{lem}

Finally, for the Hamming distance problem we show that there exists an $F$ and distribution for $U$ such that sufficiently many nodes are high-entropy nodes. The proof of the next lemma is rather involved and is given over the Sections~\ref{sec:ham} to~\ref{sec:proofvecsum}.

\begin{lem}
\label{lem:ham-H-lower}
    For the Hamming distance problem there exists a hard distribution with a fixed $F$ and random $U$ such that any node $v\in\calT$ for which $\ell_v\geq \sqrt{n}$ is a high-entropy node,
    where $h$ is a constant.
\end{lem}

In the proof of Lemma~\ref{lem:ham-H-lower} we demonstrate that there exists a very specific set of strings such that when $F$ is drawn randomly from this set, there is a non-zero probability of picking an $F$ for which many nodes are high-entropy nodes.
Unlike the convolution and multiplication problems, the distribution for $U$ is not uniform over of all strings $[q]^n$.

\subsection{Lower bounds on the information transfer}

In the previous section we gave a series of lemmas saying that for all three problems we consider, there are instances for which many nodes of $\calT$ are high-entropy nodes. In this section we combine these results with the entropy upper bound of Lemma~\ref{lem:H-upper-old} to show that many nodes have \lit.
The following lemmas match the lemmas of the previous section.
We start with the convolution problem.

\begin{lem}
\label{lem:conv-random-it}
For the convolution problem where both $F$ and $U$ are chosen uniformly at random from $[q]^n$, and $q$ is a prime, every $v\in\calT$ has \lit.
\end{lem}
\begin{proof}
By combining Lemmas~\ref{lem:H-upper-old} and~\ref{lem:conv-rand-H-lower} we have that for any $v\in\calT$ under fixed $\Xvknown$, at least half of all $F\in [q]^n$ imply that $v$ is a high-entropy node, that is,
\[
k\cdot \delta \cdot\ell_v
~\leq~
w + 2w\cdot \expected{\vit \mid \Xvknown=\Xvfix},
\]
where $k$ is the constant from Definition~\ref{def:high-node} of a high-entropy node.
Rearranging terms gives
\[
\expected{\vit \mid \Xvknown=\Xvfix}
~\geq~
\frac{\delta \cdot\ell_v}{2k\cdot w} - \frac12.
\]
We remove the conditioning by taking expectation over $\Xvknown$ under a random $U$.
When $F$ is chosen uniformly at random from $[q]^n$ we therefore have
\[
\expected{\vit}
~\geq~
\frac{\delta \cdot\ell_v}{4k\cdot w} - \frac14,
\]
hence $v$ has \lit.
\end{proof}

Similarly to Lemma~\ref{lem:conv-random-it}, we combine Lemmas~\ref{lem:H-upper-old} and~\ref{lem:conv-H-lower} to obtain the following property for the case where $F$ is a fixed string and not randomly chosen.

\begin{lem}
  \label{lem:conv-it}
  For the convolution problem there exists a hard distribution where $F$ is fixed and $U$ is chosen uniformly at random from $[q]^n$, such that every $v\in\calT$ has \lit.
\end{lem}
\begin{proof}
  Similarly to the proof of Lemma~\ref{lem:conv-random-it} we combine Lemmas~\ref{lem:H-upper-old} and~\ref{lem:conv-H-lower} to obtain, for all $v\in\calT$ under fixed $\Xvknown$,
  \[
    \expected{\vit \mid \Xvknown=\Xvfix}
    ~\geq~
    \frac{\delta \cdot\ell_v}{2k\cdot w} - \frac12,
    \]
  where $k$ is the constant from Definition~\ref{def:high-node} of a high-entropy node.
  The conditioning is removed by taking expectation over $\Xvknown$ under a random $U$.
\end{proof}

The proofs of the following two lemmas, in which we establish large information transfer for the multiplication problem, are similar to the proofs of the previous two lemmas, only that we here combine Lemma~\ref{lem:H-upper-old} with Lemmas~\ref{lem:mult-rand-H-lower} and~\ref{lem:mult-H-lower}, respectively.

\begin{lem}
\label{lem:mult-random-it}
For the online multiplication problem where both operands are chosen uniformly at random from $[q^n]$, every $v\in\calT$ has \lit.
\end{lem}

\begin{lem}
  \label{lem:mult-it}
  For the online multiplication problem there exists a fixed operand in $[q^n]$ such that when the other operand is chosen uniformly at random from $[q^n]$, every $v\in\calT$ has \lit.
\end{lem}

Finally, large information transfer is also established for the Hamming distance problem. The proof of the next lemma is identical to the proof of Lemma~\ref{lem:conv-it}, only that we
combine Lemma~\ref{lem:H-upper-old} with Lemma~\ref{lem:ham-H-lower} instead, and
restrict the nodes $v$ to those for which $\ell_v\geq \sqrt{n}$.

\begin{lem}
\label{lem:ham-it}
There exists a hard distribution for the Hamming distance problem such that every $v\in\calT$ for which $\ell_v\geq \sqrt{n}$ has \lit.
\end{lem}

\subsection{Obtaining the cell-probe lower bounds}

Now that we have established \lit for sufficiently many nodes of~$\calT$ we are ready to prove the lower bounds of Theorems~\ref{thm:conv}, \ref{thm:mult} and~\ref{thm:ham}.

For both the convolution and multiplication problems, \lit has been established for every node $v$ of $\calT$, whereas for the Hamming distance problem, \lit has only been established where $\ell_v\geq \sqrt{n}$.
In order to unify the presentation of the proofs we restrict the summation of $\vit$ to nodes for which $\ell_v\geq \sqrt{n}$.
Let $V$ denote this set of nodes.
We have
\begin{equation}
  \label{eq:total-it}
  \mathbb{E}\left[\sum_{v\in\calT} \vit\right]
  ~\geq~
  \mathbb{E}\left[\sum_{v\in V} \vit\right]
  ~=~
  \sum_{v\in V} \mathbb{E}[\vit]
  ~\geq~
  \sum_{v\in V} \frac{k\cdot\delta\cdot\ell_v}{w}
  ~=~
  \frac{k'\cdot\delta\cdot n\cdot \log n}{w},
\end{equation}
where $k$ is the constant from Definition~\ref{def:large-it} of \lit and $k'$ is a new suitable constant.
The first equality follows by linearity of expectation and the second inequality follows by Lemmas~\ref{lem:conv-random-it} to~\ref{lem:ham-it}, respectively.
The last equality follows from the fact that
\[
\sum_{\substack{v\in T\\ \ell_v \geqslant \sqrt{n}}} \ell_v ~\in~ \Theta(n\log n).
\]

Since the running time is bounded by the number of cell probes we have from Equation~(\ref{eq:total-it}) that the expected running time for any deterministic algorithm solving the convolution, multiplication or Hamming distance problem, respectively, on $n$ random inputs is
\[
  \Omega\left(\frac{\delta\cdot n\cdot \log n}{w}\right).
\]
By Yao's minimax principle, as discussed in Section~\ref{sec:preliminaries}, this implies that any randomised algorithm on its worst case input has the same lower bound on its expected running time.
The amortised time per arriving value is obtained by dividing the running time by $n$.
This concludes the proofs of Theorems~\ref{thm:conv}, \ref{thm:mult} and~\ref{thm:ham}.

\section{Hard distributions for the convolution problem} \label{sec:conv}

In this section we prove Lemmas~\ref{lem:conv-rand-H-lower} and~\ref{lem:conv-H-lower}, that is we show that there are instances to the convolution problem such that the conditional entropy of the outputs $\Yv$ is large, where all inputs but $\Xv$ are fixed.

We begin by proving Lemma~\ref{lem:conv-rand-H-lower} because the proof is straightforward and the description of the hard distribution is simple: pick the inputs $U$ uniformly at random from $[q]^n$. As to the choice of $F$ we only argue that a large fraction of all $n$-length arrays have the desired entropy lower bound. In Section~\ref{sec:conv-fixed-F} we will specify a particular $F$ with this property, which will lead to a proof of Lemma~\ref{lem:conv-H-lower}.

\subsection{Entropy lower bound over all arrays $F$} \label{sec:conv-random-F}

We now prove Lemma~\ref{lem:conv-rand-H-lower}.
Let $v$ be any internal node of $\calT$ and let $t_v\in[n]$ denote the arrival time of $\Xv[0]$.
Let $\ell=\ell_v/2$.
For $i\in [\ell]$, the $i$-th output in $\Yv$ can be broken into two sums $\calA_i$ and $\widetilde{\calA}_i$, such that $\Yv[i] = \calA_i + \widetilde{\calA}_i$, where
\begin{equation*}
    \calA_i = \sum_{j\in[\ell]} \big(F[n-1-(\ell+i)+j]\cdot U_v[j]\big)
\end{equation*}
is the contribution from the alignment of $F$ with $U_v$, and $\widetilde{\calA}_i$ is the contribution from the alignments that do not include $U_v$. Hence $\widetilde{\calA}_i$ is constant under fixed $\Xvknown$.
We define $\MFL$ to be the $\ell$$\times$$\ell$ matrix with entries $\MFL(i,j)= F[n-1-(\ell+i)+j]$. That is,
\begin{equation*}
    \MFL =
        \begin{pmatrix}
            F[n-\ell-1] & F[n-\ell+0] & F[n-\ell+1] & \cdots & F[n-2] \\
            F[n-\ell-2] & F[n-\ell-1] & F[n-\ell+0] & \cdots & F[n-3] \\
            F[n-\ell-3] & F[n-\ell-2] & F[n-\ell-1] & \cdots & F[n-4] \\
            \vdots  & \vdots & \vdots & \ddots & \vdots  \\
            F[n-2\ell] & F[n-2\ell+1] & F[n-2\ell+2] & \cdots & F[n-\ell-1]
        \end{pmatrix}.
\end{equation*}
Observe that $\MFL$ is a \emph{Toeplitz} matrix (or ``upside down'' \emph{Hankel} matrix) since it is constant on each descending diagonal from left to right.
%This property will be important later.
It follows that
\begin{equation}
    \label{eq:matrix}
    \MFL\times
        \begin{pmatrix}
            U_v[0] \\
            U_v[1] \\
            \vdots \\
            U_v[\ell-1]
        \end{pmatrix}
        =
        \begin{pmatrix}
            \calA_0 \\
            \calA_1 \\
            \vdots \\
            \calA_{\ell-1}
        \end{pmatrix}
\end{equation}
which describes a system of linear equations. Since outputs are given modulo $q$, where $q$ is assumed to be a prime, we operate in the finite field $\mathbb{Z}/q\mathbb{Z}$.
It has been shown in~\cite{KL1996:Toeplitz} that for any $\ell$, out of all the $\ell$$\times$$\ell$ Toeplitz matrices over a finite field of $q$ elements, a fraction of exactly $(1-1/q)$ is non-singular.
This fact was actually already established in~\cite{Day1960:Matrices} almost 40 years earlier but incidentally reproved in~\cite{KL1996:Toeplitz}.
%Since we have assumed in the statement of the lemma that $q$ is a prime, the ring $\mathbb{Z}/q\mathbb{Z}$ we operate in is indeed a finite field. 
Thus,
a $(1-1/q)$-fraction of all $F$ has the property that all the $\ell$ inputs in $\Xv$ can be uniquely determined from the outputs in $\Yv$.
Since the induced distribution for $\Xv$ under any fixed $\Xvknown$ is the uniform distribution on $[q]^\ell$, the conditional entropy
\[
  H(\Yv\mid\Xvknown=\Xvfix)
  ~=~
  \ell \cdot \log_2 q
  ~\geq~
  \frac{\delta \cdot \ell_v}{2},
\]
where $\delta=\lfloor\log_2 q\rfloor$.
This concludes the proof of Lemma~\ref{lem:conv-rand-H-lower}.
%\margin{Add somewhere that there is always a prime in $[q/2,q]$, hence\dots}

\subsection{Entropy lower bound with a fixed array $F$} \label{sec:conv-fixed-F}

We now prove Lemma~\ref{lem:conv-H-lower} by demonstrating that it is possible to design a fixed array $F$ such that for all nodes $v\in\calT$, a large portion of the values in $\Xv$ can be uniquely determined from the outputs $\Yv$. Since $U$ is drawn uniformly at random from $[q]^n$, this implies large entropy of the outputs $\Yv$.

The fixed array $F$ that we consider consists of stretches of~0s interspersed by~1s. The distance between two succeeding 1s is an increasing power of two, ensuring that for half of the alignments of $F$ and $S$ in the arrival interval where $\Yv$ is outputted, all but exactly one element of $U_v$ are simultaneously aligned with a 0 in $F$, hence not contributing to the outputted inner product of $F$ and $S$. We define $K_n\in[2]^n$ such that
\begin{equation*}
    K_n[0],K_n[1],\dots,K_n[n-1]\;=\;\dots000000000{\bf 1}000000000000000{\bf 1}0000000{\bf 1}000{\bf 1}0{\bf 11}0,
\end{equation*}
where commas between elements on the right hand side have been omitted, or formally,
\begin{equation*}
    %\label{eq:K}
    K_n[i] =
    \begin{cases}
        1, &\text{if $n-1-i$ is a power of two;}\\
        0, &\text{otherwise.}
    \end{cases}
\end{equation*}
The hard distribution for Lemma~\ref{lem:conv-H-lower} is $F=K_n$ and the inputs $U$ drawn uniformly at random from $[q]^n$.

Let $v$ be any node of $\calT$ and consider Figure~\ref{fig:conv} which illustrates three alignments of $F$ and $S$, denoted \alignment{1}, \alignment{2} and~\alignment{3}, respectively.
\begin{figure*}[t]
    \centering
    \insertdiagram{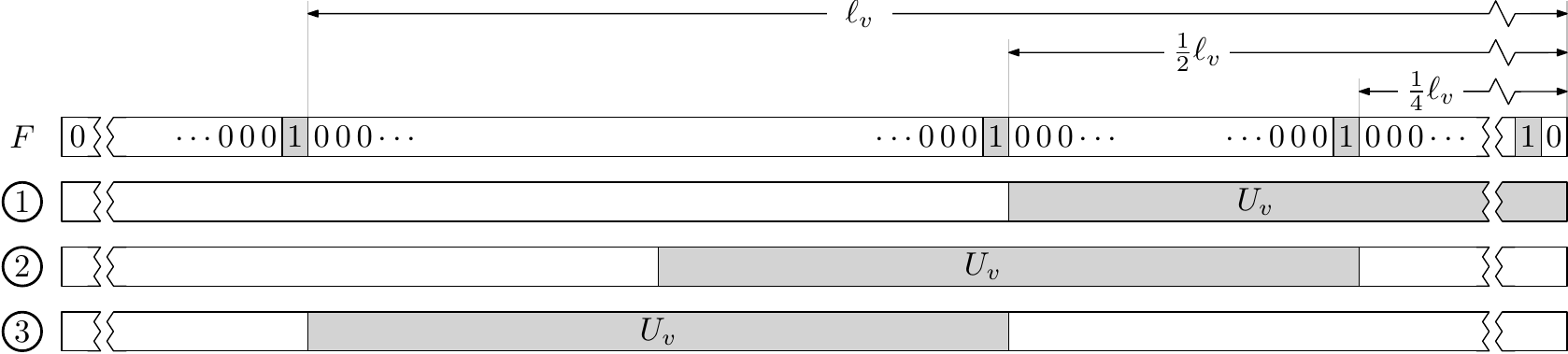}
    \caption{\label{fig:conv}Three alignments of $F=K_n$ and the stream $S$:
    \alignment{1}~the last value of $\Xv$ has just arrived, \alignment{2}~half of the outputs in $\Yv$ have been outputted, and \alignment{3}~all outputs in $\Yv$ have been outputted.}
\end{figure*}
At alignment~\alignment{1}, the last value of $\Xv$ has just arrived in the stream. At alignment~\alignment{2}, half of the outputs in $\Yv$ have been outputted.
At alignment~\alignment{3}, all outputs in $\Yv$ have been outputted.
The key observation is that between alignment~\alignment{2} and~\alignment{3}, exactly one input $x$ of $\Xv$ is aligned with a~1 in~$F$, hence $x$ can be uniquely determined from the corresponding output. Thus, over all outputs $\Yv$, a total of $\ell_v/4$ values of $\Xv$ can be determined, implying that the entropy of $\Yv$ must be at least $\delta\cdot\ell_v/4$, where $\delta=\lfloor\log_2 q\rfloor$.
We now formalise this reasoning.

Using the definition of $\ell=\ell_v/2$ and the matrix $\MFL$ above, recall that entry
$\MFL(i,j)= F[n-1-(\ell+i)+j]$.
Thus, $\MFL(i,j)=1$ if and only if
\[
  n-1-\big(n-1-(\ell+i)+j\big)~=~\ell+i-j
\]
is a power of two.
Since $\ell$ is a power of two it follows that for row $i\in \{\ell/2,\dots,\ell-1\}$ there can be at most one entry with the value~1. More precisely,
\begin{equation*}
  \MFL(i,j) =
  \begin{cases}
    1 &\textup{if $j=i$,}\\
    0 &\textup{otherwise.}
  \end{cases}
\end{equation*}
From the system of linear equations in Equation~(\ref{eq:matrix}) it follows that for $i\in \{\ell/2,\dots,\ell-1\}$,
$\calA_i = \Xv[i]$.
Since the induced distribution for $\Xv$ under any fixed $\Xvknown$ is the uniform distribution on $[q]^\ell$, the conditional entropy
\[
  H(\Yv\mid\Xvknown=\Xvfix)
  ~=~
  \frac{\ell}{2} \cdot \log_2 q
  ~\geq~
  \frac{\delta \cdot \ell_v}{4},
\]
where $\delta=\lfloor\log_2 q\rfloor$.
This concludes the proof of Lemma~\ref{lem:conv-H-lower}.

\section{Hard distributions for the multiplication problem} \label{sec:mult}

In this section we prove Lemmas~\ref{lem:mult-rand-H-lower} and~\ref{lem:mult-H-lower}, that is we show that there are instances of the online multiplication problem such that the conditional entropy of the outputs $\Yv$ is large, where all inputs but $\Xv$ are fixed.
For the purposes of proving a lower bound we assume that all digits of the operand $F$ are available at any time whereas the digits of the operand $U$ arrive one at a time.
Figure~\ref{fig:mult} illustrates $U\times F$, where $U[0]$ and $F[0]$ are the least significant digits and the product $A$ is capped at $n$ digits.
%The multiplication is in base $q$.

\begin{figure*}[t]
    \centering
    \insertdiagram{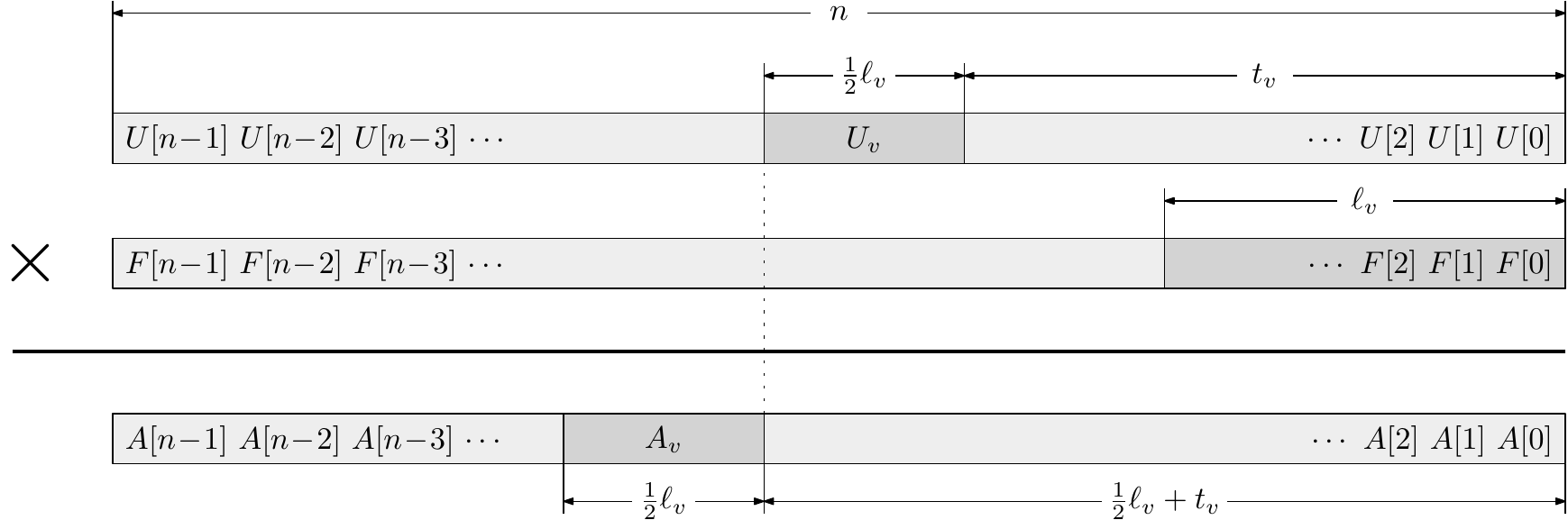}
    \caption{\label{fig:mult}An illustration of $A=U\times F$. Digits of $U$ arrive one at a time, where $U[0]$ is the low-order digit that arrives first.}
\end{figure*}

The following property of multiplying binary numbers was established by Paterson, Fischer and Meyer~\cite{PFM:1974}. The lemma is stated in our notation, but the translation from the original notation of~\cite{PFM:1974} is straightforward.

\begin{lem}[Corollary of Lemma~5 in~\cite{PFM:1974}]
    \label{lem:PFMlemma5}
        Suppose $q=2$. Let $v$ be any node of $\calT$ and fix the digits of $\Xvknown$ arbitrarily.
        At least half of all $F[0,\ell_v-1]\in [q]^{\ell_v}$ (first $\ell_v$ digits of $F$) have the property that
        any value of $A_v$ can arise from at most four distinct~$\Xv$.
\end{lem}

Although Lemma~\ref{lem:PFMlemma5} applies only to binary numbers, it naturally scales to any $q$ that is a power of two. To see this, observe that the property holds for any $v$, and a sequence of digits in base $q$ is after all just a bit sequence.

\begin{cor}
Lemma~\ref{lem:PFMlemma5} holds for any $q$ that is a power of two.
\end{cor}

We use the above corollary to prove Lemma~\ref{lem:mult-rand-H-lower}.
Let $v$ be any node of $\calT$.
At least half of all $F\in [q^n]$ have the property that $\Xv$ can be determined to up to set of four possible values given the outputs in $\Yv$.
Since the induced distribution for $\Xv$ under any fixed $\Xvknown$ is the uniform distribution on $[q]^\ell$ (the digits of $\Xv$), the conditional entropy
\[
  H(\Yv\mid\Xvknown=\Xvfix)
  ~\geq~
  \log_2 \left(\frac{q^{\ell_v / 2}}{4}\right)
  ~\geq~
  \frac{\delta \cdot \ell_v}{2} - 2,
\]
where $\delta=\log_2 q$.
This concludes the proof of Lemma~\ref{lem:mult-rand-H-lower}.

In order to prove Lemma~\ref{lem:mult-it} we specify a fixed $F$ which together with the uniform distribution for $U$ gives the desired entropy lower bound.
Similarly to the array $K_n$ from Section~\ref{sec:conv-fixed-F} we define $K_{q,n}$ to be the largest number in $[q^n]$ such that the $i$-th bit in the binary expansion of $K_{q,n}$ is $1$ if and only if $i$ is a power of two (starting with $i=0$ at the lower-order end). Thus, the binary expansion of $K_{q,n}$ is the reverse of $K_{n\log_2 q}$.
For example, suppose that $q=16$ (i.e.~hex) and $n=8$. Then $K_{16,8}=10116$ in base 16, or 65,814 in decimal, since the binary expansion of $K_{16,8}$ is
\[
 \underbrace{0000}_{0}
 \underbrace{0000}_{0}
 \underbrace{0000}_{0}
 \underbrace{0001}_{1}
 \underbrace{0000}_{0}
 \underbrace{0001}_{1}
 \underbrace{0001}_{1}
 \underbrace{0110}_{6}.
\]

Paterson, Fischer and Meyer~\cite{PFM:1974} also studied the multiplication of binary numbers where one operand is fixed. The following property was given in~\cite{PFM:1974}, here translated into our notation.

\begin{lem}[Lemma~1 of~\cite{PFM:1974}]
    \label{lem:PFMlemma1}
    Suppose $q=2$ and $F=K_{q,n}$. Let $v$ be any node of $\calT$ and fix the digits of $\Xvknown$ arbitrarily. Any value of $A_v$ can arise from at most two distinct~$\Xv$.
\end{lem}

Similarly to Lemma~\ref{lem:PFMlemma5} and from our definition of $K_{q,n}$, the above lemma scales to any $q$ that is a power of two.

\begin{cor}
Lemma~\ref{lem:PFMlemma1} holds for any $q$ that is a power of two.
\end{cor}

We use the above corollary to prove Lemma~\ref{lem:mult-H-lower} where $F=K_{q,n}$.
Let $v$ be any node of $\calT$.
The value of $\Xv$ can be determined to up to set of two possible values given the outputs in $\Yv$.
Since the induced distribution for $\Xv$ under any fixed $\Xvknown$ is the uniform distribution on $[q]^\ell$ (the digits of $\Xv$), the conditional entropy
\[
  H(\Yv\mid\Xvknown=\Xvfix)
  ~\geq~
  \log_2 \left(\frac{q^{\ell_v / 2}}{2}\right)
  ~\geq~
  \frac{\delta \cdot \ell_v}{2} - 1,
\]
where $\delta=\log_2 q$.
This concludes the proof of Lemma~\ref{lem:mult-H-lower}.

\section{Hard distribution for the Hamming distance problem} \label{sec:ham}

In this section we prove Lemma~\ref{lem:ham-H-lower}, that is we show that there are instances of the Hamming distance problem such that the conditional entropy of the outputs $\Yv$ is large, where all inputs but $\Xv$ are fixed.
We will show this property for nodes in the upper part of the tree~$\calT$, namely nodes $v$ such the number of leaves $\ell_v$ is greater than some constant times~$\sqrt{n}$.

Unlike the hard distributions we gave for the convolution and multiplication problems, we will not give an explicit description of the array $F$ for which the Hamming distance lower bound holds. We only show the existence of such an $F$.
Further, for both the convolution and multiplication problems we showed that the lower bound was obtained for a majority of all $F$, where $U$ was chosen uniformly at random from $[q]^n$.
 For the Hamming distance problem we will instead show that there exists an $F$ and some particular subset of $[q]^n$ such that when $U$ is drawn uniformly at random from this subset, we obtain the desired lower bound.

\subsection{Terminology, choice of $q$ and rounding issues} \label{sec:rounding}

We will refer to the input arrays, including $F$ and $U$, as \emph{strings}, and the set $[q]$ as the \emph{alphabet}. The values of the alphabet are referred to as \emph{symbols}.

Unlike the convolution and multiplication problems, for the Hamming distance problem there is no benefit in having an alphabet size greater $n$, the length of $F$.
Our hard distribution is constructed such that with an alphabet of size $q$, $n$ has to be roughly $q^3$, or more.
So from now on we assume that~$n\geq q^3$.
Observe that whenever $n$ is polynomial in $q$, the number of bits needed to represent a symbol is $\delta\in\Theta(\log n)$.

We will introduce two special symbols denoted $\psymb$ and $\tsymb$. It will be tidy to keep them separate throughout the presentation. Once we start digging into the details we will see that for a given $q$, the number of distinct symbols that we actually use in the hard instance is only $q-\sqrt{q}+2$, including the two special symbols. The alphabet $[q]$ is therefore large enough to accommodate every symbol that we use.

We will often treat various roots of integers as integers. For example, we may say that some string of length $q^{3/2}$ is the concatenation of $q$ smaller strings, each of length $q^{1/2}$. This is of course only possible whenever these numbers are integers, which is not necessarily the case for arbitrary $q$. One could overcome this problem by adjusting the values with appropriate floors and ceilings, as well as introducing padding symbols where necessary, but this would without doubt clutter the presentation. We have decided to keep it simple by treating any root of any integer as an integer, and assuming that everything adds up nicely. This is only to keep the presentation clean and it should be obvious from the context that this has no impact on the asymptotic behaviour.

\subsection{The overall structure of the fixed string $F$}

Recall the definition of the array $K_n\in\{0,1\}^n $ from Section~\ref{sec:conv-fixed-F} which consists of 0s everywhere except for at power-of-two positions from the right-hand end.
A hard distribution for the convolution problem was given by setting $F$ to $K_n$ and choosing $U$ uniformly at random from $[q]^n$.
Recall Figure~\ref{fig:conv} which illustrates why we chose this hard distribution: for each output in the second half of $A_v$,
that is between the alignments marked~\alignment{2} and~\alignment{3} in the figure,
exactly one input of $U_v$ is aligned with a 1 in $F$ and all other inputs of $U_v$ are aligned with 0. Thus, the second half of $\Xv$ can be uniquely determined from the outputs $\Yv$.

To show a lower bound for the Hamming distance problem we devise a string $F$ that resembles $K_n$.
First we introduce an auxiliary string $R$ of length $q^{3/2}$. We will use $r=q^{3/2}$ as a shorthand for $|R|$. Recall that $n\geq q^3$.
We will give the details of $R$ later but will highlight an important property of it below.
We obtain $F$ from $K_n$ by first replacing each 0 by a symbol that we denote~$\psymb$. The symbol~$\psymb$ will never occur in the stream, hence will always generate a mismatch.
We then replace every $r$-length substring starting at a~1 with a copy of $R$. Any~1 that is closer than $r$ positions from the right-hand end of $F$ is replaced by a $\psymb$-symbol instead.
Figure~\ref{fig:fstring} illustrates $F$.
\begin{figure}[t]
    \centering
    \insertdiagram{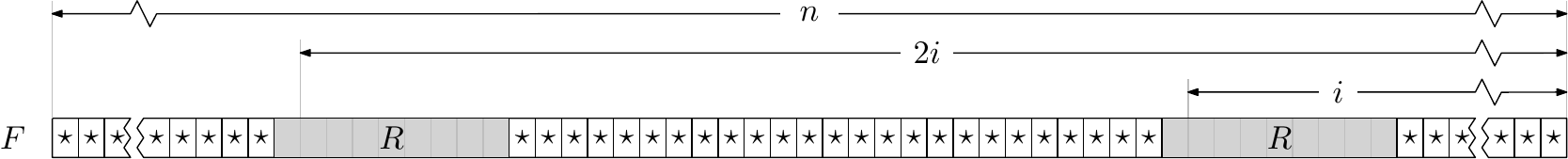}
    \caption{\label{fig:fstring}The string $F$ has a copy of $R$ starting at each position $n-1-i$ where $i\geq |R|$ is a power of two. All other positions have the symbol~$\psymb$ which only occurs in $F$ and not in the stream.}
\end{figure}

\subsection{Properties of the string $R$ and Hamming arrays}

The string $R$ will play the same role as the value~1 in $K_n$ did for the convolution problem, namely it will allow us to uniquely determine symbols from $U$.
To see how, we first introduce the notion of a Hamming array, illustrated in Figure~\ref{fig:hamarray}.
\begin{figure*}[t]
    \centering
    \insertdiagram{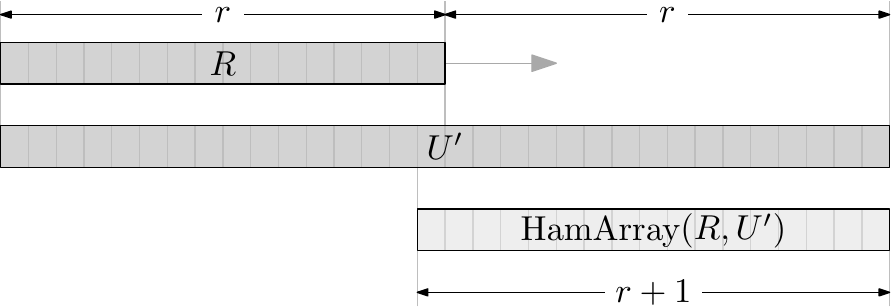}
    \caption{\label{fig:hamarray}$\Hamarray(R,U')$ contains the Hamming distances between $R$ and every $r$-length substring of $U'$ as $R$ slides along $U'$.}
\end{figure*}
For a string $U'$ of length $2r$,
we write $\Hamarray(R,U')$ to denote the $(r+1)$-length array such that for $i\in[r+1]$, $\Hamarray(R,U')[i]$ is the Hamming distance between $R$ and $U'[i,i+r-1]$. That is, $\Hamarray(R,U')$ contains the Hamming distances between $R$ and every $r$-length substring of $U'$. 

To see the resemblance with a 1 in $K_n$, we give the following lemma. The proof is non-trivial and deferred to Section~\ref{sec:R}. A high-level explanation of the lemma is given immediately after its statement.

\begin{lem}
\label{lem:combinatorial}
    There exists a constant $k>0$ such that for any $r$ there is an $r$-length string $R \in [r^{2/3}]^r$ such that
    \[
      \Big|\Set{\Hamarray(R,U') \;\mid\; \textup{$U' \in [r^{2/3}]^{2r}$}}\Big| \,\geq\, r^{kr}.
    \]
\end{lem}

First recall that $q=r^{2/3}$, hence both $R$ and $U'$ of the lemma are over an alphabet of $q$ symbols.
The lemma says that there is a string $R$ such that over all possible $U'$ of length $2|R|$, one can obtain $q^{\Theta(r)}$ distinct Hamming arrays. Since there are only $q^{2r}$ possible values of $U'$, this is means that a non-negligible fraction of all $U'$ can be put in one-to-one correspondence with Hamming arrays.
Thus, as symbols in $\Xv$ slide past an $R$ in a similar fashion to symbols in $\Xv$ sliding past a 1 in $K_n$ in the hard distribution for the convolution problem, we can infer a substantial portion of the symbols of $\Xv$ from the outputs $\Yv$, hence obtain large entropy.
We formalise this in the next section and explain how the lower bound is obtained.

\subsection{The hard distribution and obtaining the lower bound}

Relying on Lemma~\ref{lem:combinatorial} above we will now describe a hard distribution for the Hamming distance problem and use it to prove Lemma~\ref{lem:ham-H-lower}.
Given a string $R \in [q]^r$, we let
\[
  \calU_R \subseteq [q]^{2r}
\]
be any largest set of $2r$-length strings such that for any two distinct strings $U'_1, U'_2\in \calU_R$,
\[
  \Hamarray(R,\U'_1)\,\neq\, \Hamarray(R,\U'_2).
\]
To uniquely specify a string in $\calU_R$ we need $\log_2|\calU_R|$ bits.
By Lemma~\ref{lem:combinatorial} we have that there exists an $R$ such that $\log_2|\calU_R|\in \Theta(r\log q)$ since $q=r^{2/3}$.

For the hard distribution we use $F$ from above with an $R$ that has the properties of Lemma~\ref{lem:combinatorial}. The input $U$ is given by concatenating $n/2r$ strings drawn independently and uniformly at random from $\calU_R$.

Similarly to Figure~\ref{fig:conv} we can now illustrate how strings from $\calU_R$ slide past $R$ during the second half of the outputs in $\Yv$, where $v$ is any node of $\calT$ such that $\ell_v\geq \sqrt{n} \geq r$. Recall that we have assumed that $n\geq q^3=r^2$.
In Figure~\ref{fig:doubling} we have illustrated $U_v$ as the concatenation of random strings $U'_1,\dots,U'_m$ drawn from $\calU_R$, where $m=\ell_v/(4r)$.
\begin{figure*}[t]
    \centering
    \insertdiagram{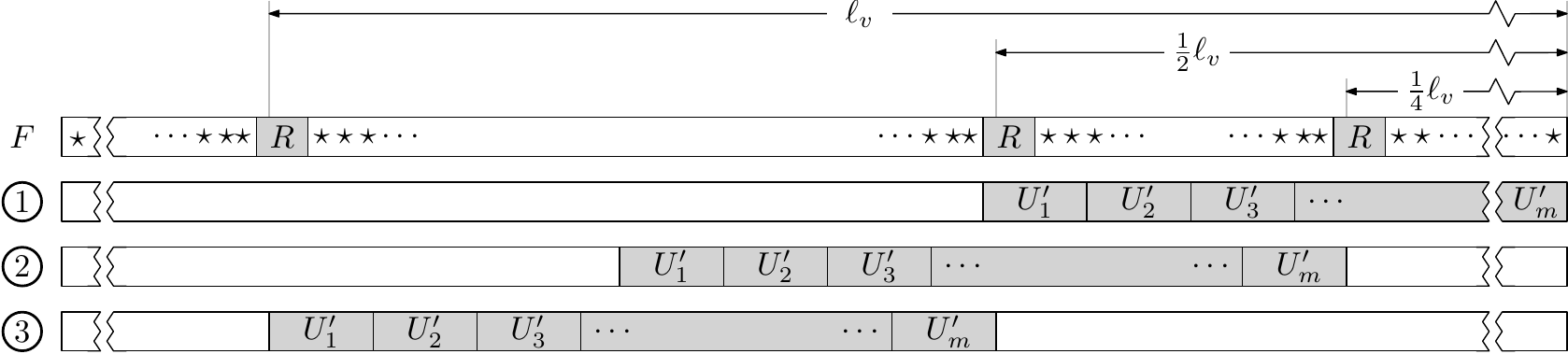}
    \caption{\label{fig:doubling}
    Three alignments of $F$ and the stream $S$:
    \alignment{1}~the last value of $\Xv$ has just arrived, \alignment{2}~half of the outputs in $\Yv$ have been outputted, and \alignment{3}~all outputs in $\Yv$ have been outputted. The string $\Xv$ is here the concatenation of $U'_1,\dots,U'_m\in\calU_R$, where $m=\ell_v/(4r)$.}    
\end{figure*}
Between alignments~\alignment{2} and~\alignment{3} in the figure, the second half of the substrings $U'_i$ of $\Xv$ slide in turn past $R$, and from the outputs in $\Yv$ we can infer $\Hamarray(R,U'_i)$ for each such $U'_i$. By construction of $\calU_R$ this allows us to uniquely determine the strings $U'_i$.
Thus, over all outputs $\Yv$, a total of $m/2$ (give or take a constant number to compensate for border cases) substrings $U'_i$ of $\Xv$ can be determined, implying that the entropy of $\Yv$ must be at least, by Lemma~\ref{lem:combinatorial},
$\Theta((m/2)\cdot r\log q)=\Theta(\ell_v \cdot\delta)$,
where $\delta=\lfloor\log_2 q\rfloor$.
This concludes the proof of Lemma~\ref{lem:ham-H-lower}.

\section{A string with many~different~Hamming~arrays}

In this section we prove Lemma~\ref{lem:combinatorial},
that is we show that there exists a string $R$ which gives many different Hamming arrays.
This is arguably the most technically detailed part of our lower bound proofs.
To recap, we claim that for any $r$ there exists a string $R \in [r^{2/3}]^r$ which permits at least $r^{kr}$ distinct Hamming arrays when combined with every string in $[r^{2/3}]^{2r}$, where $k$ is a constant.
Next we describe the overall structure of an $R$ with this property.

\subsection{The structure of $R$}

To shorten notation it will be convenient to introduce the variable $\mu$ as a shorthand for $r^{1/3}$. Hence $R$ has length $r=\mu^3$ and $q=\mu^2$.  The string $R$ is constructed by concatenating $\mu^2$ substrings, each of length $\mu$.
For $i\in[\mu^2]$ we let $\ps_i$ denote the $i$-th substring of $R$, that is
\[
 R\,=\,\ps_0\,\ps_1\cdots\ps_{(\ELL^2-1)}.
\]
Each substring $\ps_i$ can only contain symbols from the set $\{\psymb,i\}$, where $\psymb$ is the special symbol that will not occur in the stream.
Figure~\ref{fig:R} illustrates an example of $R$.
\begin{figure*}[t]
    \centering
    \insertdiagram{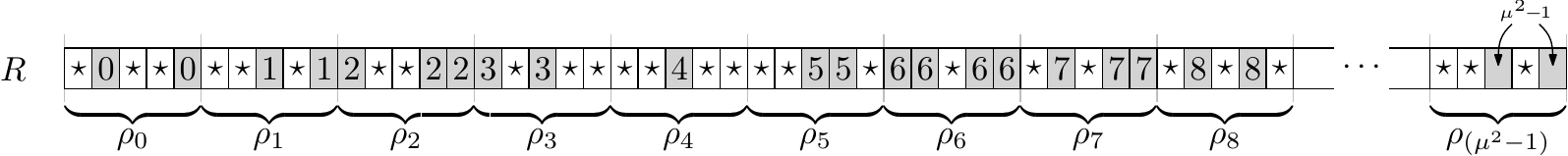}
    \caption{\label{fig:R}An example of the string $R$ of length $r=\mu^3$, which is the concatenation of the $\mu^2$ strings $\ps_0,\dots,\ps_{\mu^2-1}$, where $\ps_i\in \{\psymb,i\}^\mu$.}    
\end{figure*}

Doing the maths correctly, the total number of distinct symbols in $R$ could reach $\mu^2+1=r^{2/3}+1=q+1$. As pointed out in Section~\ref{sec:rounding} we do indeed introduce two additional symbols, of which one is $\psymb$, however, to keep notation clutter-free we abuse the notion of $q$ by giving it a slack that should obviously be adjusted by some constant where appropriate.

The purpose of the substrings $\ps_i$ is to support a reduction from vector addition to Hamming arrays that we explain next.

\subsection{Vector sums and Hamming arrays} \label{sec:vecsums-hamarrays}

The $\mu$-length substring $\ps_i$ of $R$ corresponds to a 0/1-vector $v_i\in\{0,1\}^\mu$ such that the $j$-th component of $v_i$ is 0 if and only the $j$-th symbol of $\ps_i$ is $\psymb$. For example, $\ps_2={2}{\star}{\star}{2}{2}$ from Figure~\ref{fig:R} corresponds the vector $v_2=(1,0,0,1,1)$.

To explain the idea of how vector addition can be carried out by using the concept of a Hamming array of $R$ and some string $U'\in[\mu^2]^{2r}$, consider Figure~\ref{fig:blocks} as an illustrative example.
\begin{figure*}[t]
    \centering
    \insertdiagram{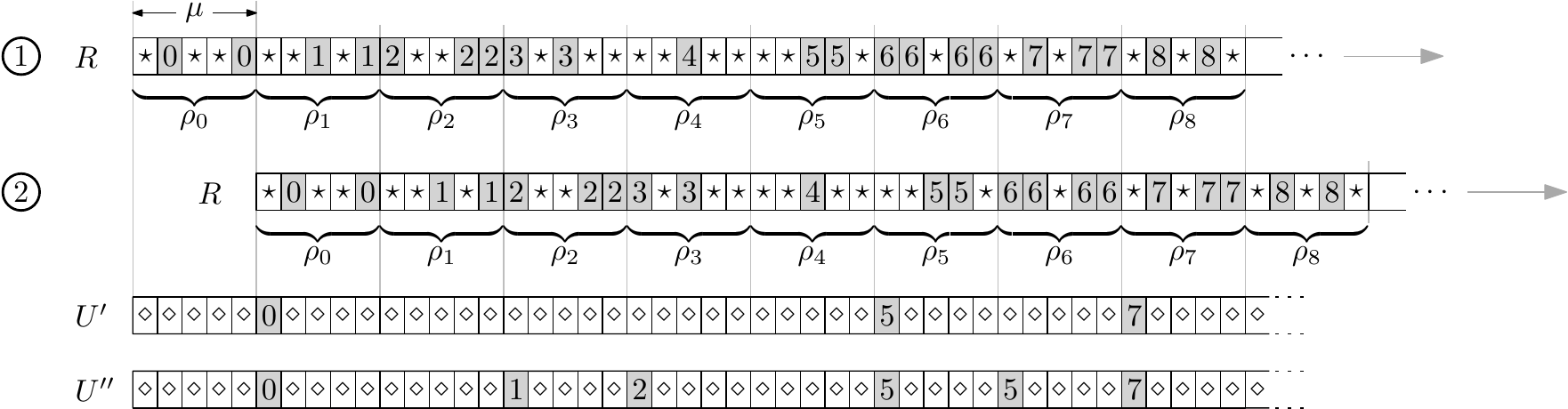}
    \caption{\label{fig:blocks}Setting symbols of $\U'$ renders a large set of possible Hamming distance outputs.}
\end{figure*}
Here the string $U'$ contains the other special symbol that we introduce, denoted~$\tsymb$. This symbol does not occur in $R$, hence will always mismatch.
In the figure we see that all positions of $U'$ have the symbol $\tsymb$, except for three positions where the symbols are 0,~5 and~7, respectively. The positions holding these symbols are chosen such that in the first alignment between $R$ and $U'$, marked~\alignment{1}, the symbols 0,~5 and~7 sit immediately after $\ps_0$, $\ps_5$ and $\ps_7$ in $R$, respectively.
As $R$ slides $\mu$ steps to the right towards the alignment marked~\alignment{2}, the symbols 0, 5 and~7 of $U'$ will generate matches whenever they are aligned with their corresponding symbols in $R$.
Thus, for $i\in\{1,\dots,\mu\}$,
\[
    \Hamarray(R,U')[i] ~=~ r - (v_0 + v_5 + v_7)[i],
\]
where $(v_0 + v_5 + v_7)[i]$ is the $i$-th component of the sum of the vectors $v_0$, $v_5$ and~$v_7$.
In other words, from $\Hamarray(R,U')[1,\mu]$ we can uniquely determine the sum $v_0+v_5+v_7$.

The idea above can be repeated by populating $U'$ with more symbols from $[\mu^2]$. As an example we have added the symbols 1 and 2, and another copy of~5 to $U'$, which is the string denoted $U''$ in the figure.
As $R$ slides another $\mu$ steps to the right, $\Hamarray(R,U'')[\mu+1,2\mu]$ uniquely specifies the sum $v_1+v_2+v_5$.

Observe that as we populate $U'$ with symbols, positions get \emph{blocked}. For example, we cannot obtain the sum $v_1+v_2+v_4$ from $\Hamarray(R,U'')[\mu+1,2\mu]$ since the position where the~4 has to be set is already occupied by a~5.
Observe however that setting symbols of $U'$ as above generates matches only in the intended $\mu$-length window of the Hamming array. Thus, we have full control of which vector sums we want to compute, under the constraint that positions get blocked, limiting the choice of vectors.

The conclusion this far is that vector sums have a direct correspondence with the Hamming array. Next we take the ideas from above further and show that if there exists a pool of $\mu^2$ vectors such that many different vector sums can be obtained when adding $\mu$ vectors from the pool, then the number of distinct $\Hamarray(R,U')$ one can obtain is large. This would prove Lemma~\ref{lem:combinatorial}.

\subsection{The string $R$ and the proof of Lemma~\ref{lem:combinatorial}} \label{sec:R}

Before we state the next lemma we need to define what we mean by sub-multiset of a multiset~$X$. We consider an arbitrary ordering of the elements of $X$ and refer to $X[i]$ as the $i$-th element of $X$.  
We use the term \emph{sub-multiset} of $X$ to denote any multiset obtained from $X$ by removing zero or more elements. We will use the notation $\sqsubseteq$ to denote the sub-multiset relation so that we have, for example,  $\{1,1,4,5,5\} \sqsubseteq \{1,1,1,4,4,5,5,7,8\}$.

\begin{lem}
    \label{lem:vecsum}
    For any $\ELL>40$ such that $\ELL-1$ is a prime, there exists a multiset $V$ of vectors from
    $\{0,1\}^\ELL$ such that $|V|=\ELL(\ELL-1)$ and for any sub-multiset $V'\subseteq V$ of size at least
    $(63/64)|V|$,
    \begin{align*}
        \left|\set{w_1 + \cdots + w_\ELL \,|\, \{w_1,\dots,w_\ELL\} \sqsubseteq V'}\right| ~\geq~ \ELL^{(\ELL/10)} .
    \end{align*}
\end{lem}

The lemma is proved in Section~\ref{sec:proofvecsum} and we will now use it to construct an $R$ that proves Lemma~\ref{lem:combinatorial}. The introduction of a sub-multiset $V'$ in the lemma above is to reflect the fact that positions of $U'$ get blocked as we populate it with symbols. We will see next that at any step, a fraction of at most $1/64$ of the $\mu^2$ vectors are blocked.

Suppose that $V=\{v_0,\dots,v_{\mu(\mu-1)}\}$ is a multiset of $\mu$-length vectors over $\{0,1\}$ with the properties of Lemma~\ref{lem:vecsum}. That is, we assume that $\mu>40$ and $\mu-1$ is a prime. Again as discussed in Section~\ref{sec:rounding}, we can always tweak relevant values in order to meet this criteria.

The string $R$ is simply chosen such that for $i\in[\mu(\mu-1)]$, the substring $\ps_i$ corresponds to the vector $v_i$ of $V$. For $i\in\{\mu(\mu-1),\dots,(\mu^2-1)\}$, the substring $\ps_i=\{\psymb\}^\mu$ as we will ignore these substrings anyway.
In order to show that this $R$ proves Lemma~\ref{lem:combinatorial} we will populate a $2r$-length vector $U'$ with symbols and show how $\mu$-length subarrays of $\Hamarray(R,U')$ correspond to vector sums of $\mu$ vectors chosen arbitrarily from a sub-multiset of $V$. The string $U'$ is obtained as follows:

\begin{enumerate} 
 \item Set all $2\mu^3$ positions of $U'$ to the symbol $\tsymb$.
 
 \item Align $R$ with the left half of $U'$ as illustrated in Figure~\ref{fig:hamarray}.
 
 \item Let $V'\subseteq V$ be the set of vectors that are not blocked. (Initially this means that $V'=V$ but as we return to this step, $V'$ shrinks.)
 
 \item Choose any sub-multiset $\{w_1,\dots,v_\mu\} \sqsubseteq V'$ and set their corresponding positions in $U'$ accordingly.
 
 \item Slide $R$ by $\mu$ steps along $U'$. Over these alignments, $\Hamarray(R,U')$ uniquely specify the vector sum $w_1+\cdots +w_\mu$.
 
 \item[~] Steps 3--5 are referred to as a \emph{round}.
 
 \item Repeat from Step~3 for a total of $(\mu-1)/64$ rounds. Observe that a total of $\mu(\mu-1)/64=(1/64)|V|$ vectors get blocked, hence $|V'|$ is always at least $(63/64)|V|$.
 
 \item Slide $R$ by \emph{one single step} along $U'$. This will offset all previously blocked vectors and allow us to start over again at Step~3 as if no vectors are blocked.
 This is repeated until this step is reached for the $\mu$-th time. At that point the offsetting of blocked vectors has cycled and previously set positions of $U'$ are yet again blocking.
\end{enumerate}

Populating $U'$ according to the procedure above means that $R$ is shifted by a total of
\[
  \mu\cdot (\mu-1)/64 \cdot \mu + (\mu-1) ~=~ \mu^3/64 - \mu^2/64 +\mu-1 ~<~ r
\]
steps. Over these steps we have by Lemma~\ref{lem:vecsum} that for each $\mu$-length subarray of $\Hamarray(R,U')$ that corresponds to a vector sum, there is a choice of at least $\mu^{(\mu/10)}$ distinct values.
Thus, when $\mu>40$, the number of distinct $\Hamarray(R,U')$ is at least
\[
    \left(\mu^{(\mu/10)}\right)^{ \mu(\mu-1)/64 }
    ~=~
    \mu^{(\mu^3-\mu^2)/640}
    ~\geq~
    \mu^{(\mu^3/656)}
    ~=~
    \left(r^{(1/3)}\right)^{(r/656)}
    ~=~
    r^{kr},
\]
where $k=1/1968$.
This concludes the proof of Lemma~\ref{lem:combinatorial}.

\section{Vector sets with many distinct sums} \label{sec:proofvecsum}

In this section, we prove Lemma~\ref{lem:vecsum}.
We first rephrase the lemma slightly by introducing some notation.
For any multiset $V'$ of vectors from $\{0,1\}^\ELL$, we define
\[
    \Vsum(V') = \set{w_1 + \cdots + w_\ELL \,|\, \{w_1,\dots,w_\ELL\} \sqsubseteq  V'}
\]
to be the set of distinct vector sums one can obtain by summing the vectors of $\mu$-sized sub-multisets of $V'$
Addition is element-wise and over the integers.
Lemma~\ref{lem:vecsum} says that there exists a multiset $V$ of vectors from     $\{0,1\}^\ELL$ such that $|V|=\ELL(\ELL-1)$ and for any sub-multiset $V'\sqsubseteq V$ of size at least $(63/64)|V|$, we have that $|\Vsum(V')| \geq \ELL^{(\ELL/10)}$.

Our approach will be an application of the probabilistic method. Specifically, we will show that when the vectors of $V$ are sampled uniformly at random, the expected value
\[
   \expect{\Vsum(V)}\,\geq\, \frac12 (\ELL-1)^{(\ELL/9)}.
\]
Thus, there must exist a $V$ such that $\Vsum(V)\geq (\ELL-1)^{(\ELL/9)}/2$.
Given such a $V$, we then show that for every sub-multiset $V'\sqsubseteq  V$ such that $|V'| \geq (63/64)|V|$,
$\Vsum(V')\geq \ELL^{(\ELL/10)}$.

\subsection{Vectors and codes}

We now describe a connection between vectors and codes. We will require the following lemma from the field of Coding Theory.
The lemma is tailored for our needs and is a special case of ``Construction~II'' in~\cite{AGM:1992}.
For our purposes, a binary constant-weight cyclic code can be seen simply as set of bit-strings (codewords) with two additional properties: the first is that all codewords have constant Hamming weight $\ELL$, i.e. they have exactly $\ELL$ 1s, and the second property is that any cyclic shift of a codeword is also a codeword.

\begin{lem}[\cite{AGM:1992}]
    \label{lem:cyclic-code}
    For any $\ELL\geq 4$ such that $\ELL-1$ is a prime and any odd $\gamma\in[\ELL]$, there is a binary constant-weight cyclic code with $(\ELL-1)^{\gamma}$ codewords of length $\ELL(\ELL-1)$ and Hamming weight $\ELL$ such that any two codewords have Hamming distance at least $2(\ELL-\gamma)$.
\end{lem}

Let $\Cbig$ be the binary code that contains \emph{all} codewords of length $\mu(\mu-1)$ with Hamming weight $\ELL$.
We can think of a codeword of $\Cbig$ representing a $\mu$-sized sub-multiset $X\sqsubseteq V$ such that the $i$-th vector of $V$ (under any enumeration of the elements of $V$) is in $X$ if and only if position $i$ of the codeword is~$1$.
That is, $\Cbig$ represents all possible sub-multisets of $V$ of size $\mu$.
To shorten notation, we refer to $\cbig\in\Cbig$ as both a codeword and a sub-multiset of $\mu$ vectors from $V$.

Suppose that $\mu\geq 4$ and $\mu-1$ is a prime.
We let $C\subseteq \Cbig$ be a cyclic code of size $(\ELL-1)^\gamma$,
where $\gamma$ is any odd integer in the interval $[\ELL/9,\ELL/8]$,
such that the Hamming distance between any two codewords in $C$ is at least $7\mu/4$.
The existence of such a $C$ is guaranteed by Lemma~\ref{lem:cyclic-code} since
$2(\mu-\mu/8)=7\mu/4$.
Observe that every codeword of $C$ has Hamming weight $\ELL$.

For $c\in C$ we define the \emph{ball}
\[
    \ball{c}=\set{\cbig ~\mid~ \text{$\cbig\in\Cbig$ and
    Hamming distance between $c$ and $\cbig$ is at most $\ELL/16$}}  
\]
to be the set of bit strings in $\Cbig$ at Hamming distance at most    $\ELL/16$ from~$c$.
Hence the $|C|$ balls are all disjoint since the Hamming distance between any two codewords in $C$ is at least than $7\mu/4$.
We have that for any $c\in C$, using the fact $\binom a b \leq (ae/b)^b$,
\begin{equation*}
        \label{eq:ball-size}
        \big|\ball{c}\big| \,\leq\, \binom \ELL {\ELL/16}\cdot \binom{|V|} {\ELL/16}
            \,\leq\, \left(\frac{\ELL e\cdot|V|e}{(\ELL/16)^2}\right)^{\ELL/16}
            \,\leq\, \left(\frac{\ELL}{16}\right)^{\ELL/16} .
\end{equation*}

For $\cbig\in \Cbig$ we write $\vsum(\cbig)$ to denote the vector in $[\ELL+1]^\ELL$ obtained by
adding the $\ELL$ vectors in the vector set $\cbig$, that is $\vsum(\cbig)$ vector sum of the vectors represented by $\cbig$.

Towards proving Lemma~\ref{lem:vecsum} we will show that when the vectors of $V$ are chosen uniformly at random, we expect more than half of all $|C|$ balls to have the property that for every $\cbig$ in the ball, $\vsum(\cbig)$ can only be obtained by summing vectors from that ball.

\subsection{Choosing the vectors in $V$}

So far we have not discussed the choice of vectors in $V$. 
We consider the case where the vectors are chosen independently and uniformly at random from $\{0,1\}^\ELL$.
We will first show that
\[
   \expect{\Vsum(V)}\,\geq\, \frac12 (\ELL-1)^{(\ELL/9)},
\]
then we will fix $V$ and show that it has the property of Lemma~\ref{lem:vecsum}.

For any $\cbig_1\in\ball{c_1}$ and $\cbig_2\in\ball{c_2}$, where $c_1,c_2\in C$ are distinct, we now analyse the probability that
$\vsum(\cbig_1)=\vsum(\cbig_2)$. From the definitions above it follows that $\cbig_1$ and $\cbig_2$ must differ
on at least $7\ELL/4-2(\ELL/16)\geq \ELL$ positions, implying that the two vector sets $\cbig_1$ and
$\cbig_2$ have at most $\ELL/2$ vectors in common, thus at least $\ELL/2$ of the vectors in $\cbig_1$ are not in $\cbig_2$.
Let $w_1,\dots,w_{(\ELL/2)}$ denote an arbitrary choice of $\mu/2$ of those vectors.
For $i\in[\mu]$ we can write the $i$-th component of $\vsum(\cbig_1)$ as
\[
    \vsum(\cbig_1)[i] ~=~ w_1[i]+\cdots +w_{(\ELL/2)}[i] + x[i],
\]
where the vector $x$ does not depend on $w_1,\dots,w_{(\ELL/2)}$.
In order to have $\vsum(\cbig_1)=\vsum(\cbig_2)$ we must have
\[
    w_1[i]+\cdots +w_{(\ELL/2)}[i] ~=~ \vsum(\cbig_2)[i] - x[i]
\]
for each $i\in [\ELL]$.
Since the vectors are picked independently and uniformly at random from $\{0,1\}^\mu$,
the most likely value of $w_1[i]+\cdots +w_{(\ELL/2)}[i]$ is $\ELL/4$.
The probability that this sum equals $\mu/4$ is
\[
  \Prob\left(w_1[i]+\cdots +w_{(\ELL/2)}[i] = \frac\pi 4\right)
  ~=~
    \binom{\ELL/2}{\ELL/4 }\cdot \left(\frac12\right)^{\ELL/2}
    \,\leq~
    \left(\frac \ELL 2\right)^{-1/2},
\]
where the inequality follows from the fact that for any $a$, $\binom a {a/2}\leq 2^a/\sqrt{a}$.
Thus, the probability that $\vsum(\cbig_1)=\vsum(\cbig_2)$, that is
$\vsum(\cbig_1)[i]=\vsum(\cbig_2)[i]$ for all $i\in [\mu]$, is
    \begin{equation}
        \label{eq:sum-equal}
        \Prob\big(\vsum(\cbig_1)=\vsum(\cbig_2)\big)
        \, \leq \,
        \left(\left(\frac{\ELL}{2}\right)^{-1/2}\right)^\mu
        \, \leq \,
        \left(\frac{\ELL}{2}\right)^{-\ELL/2}.
    \end{equation}
    
For two distinct $c_1,c_2\in C$, we define the indicator random variable
\begin{equation*}
  I(c_1,c_2) =
  \begin{cases}
   0 & \textup{if $\vsum(\cbig_1)=\vsum(\cbig_2)$ for some $\cbig_1\in\ball{c_1}$ and  $\cbig_2\in\ball{c_2}$,}\\
   1 & \textup{otherwise.}
  \end{cases}
\end{equation*}
Taking the union bound over all $\cbig_1\in\ball{c_1}$ and
$\cbig_2\in\ball{c_2}$, and using the probability bound in Equation~(\ref{eq:sum-equal}), we have
    \begin{align}
        \label{eq:union-inner}
        \Prob \big( I(c_1, c_2) = 0 \big)
        \, &\leq \, \big|\ball{c_1}\big|\cdot \big|\ball{c_2}\big| \cdot \left(\frac{\ELL}{2}\right)^{-\ELL/2}  \\
        \, &\leq \, \left(\frac{\ELL}{16}\right)^{2(\ELL/16)} \left(\frac{\ELL}{2}\right)^{-\ELL/2} \notag \\
        &\leq \, \left(\frac{1}{\ELL^3}\right)^{\ELL/8}. \notag
    \end{align}

For any $c_1\in C$, we now define the indicator random variable
\begin{equation*}
  I'(c_1) =
  \begin{cases}
   0 & \textup{if $I(c_1,c_2)=0$ for some $c_2\in C\setminus \{c_1\}$,}\\
   1 & \textup{otherwise.}
  \end{cases}
\end{equation*}
That is, $I'(c_1)=1$ if and only if
$\vsum(\cbig_1)\neq\vsum(\cbig_2)$
for every $\cbig_1\in\ball{c_1}$ and every $\cbig_2$ from another ball.
In other words, the sums of codewords in $\ball{c_1}$ are unique for this ball.
We say that $\ball{c_1}$ is \emph{good} if and only if $I'(c_1)=1$.
It is possible however that $\vsum(\cbig_1)=\vsum(\cbig_2)$ if $\cbig_2$ is from the same ball as $\cbig_1$ though.

Taking the union bound over all $c_2\in C$, and using Equation~(\ref{eq:union-inner}) and the fact that $|C|\leq \ELL^{(\ELL/8)}$, we have
    \begin{align*}
        %\label{eq:union-outer}
        \Prob \big( I'(c_1) = 0 \big) \, &\leq \!\!\! \sum_{c_2 \in C\setminus \{c_1\}} \!\!\!\! \Prob \big( I(c_1, c_2) = 0 \big)\\
        \, &\leq \, |C| \left(\frac{1}{\ELL^3}\right)^{\ELL/8}
        \, \leq \, \ELL^{(\ELL/8)} \left(\frac{1}{\ELL^3}\right)^{\ELL/8}
        \, \leq \, \frac{1}{2} \,. \notag
    \end{align*}
By linearity of expectation we have that the expected number of good balls is
\[
 \expect{\sum_{c\in C}I'(c)} ~\geq~ \frac{|C|}2.
\]
The conclusion is that there is a multiset $V$ of vectors for
which at least $|C|/2$ balls are good, hence
\[
  \Vsum(V)\,\geq\, \frac12 (\ELL-1)^{(\ELL/9)}
\]
since $|C|\geq (\ELL-1)^{(\ELL/9)}$.

\subsection{Many distinct sums for subsets of $V$}

Suppose now that $V$ is a multiset such that the number of good balls is at least $|C|/2$, hence $\Vsum(V)\geq (\ELL-1)^{(\ELL/9)}/2$. From the conclusion above we know that such a set must exist.
It remains to show that for any sub-multiset $V' \sqsubseteq V$ of size $(63/64)|V|$, $\vsum(V')$ is also large.

Over all codewords in $C$, seen as bit strings, the total number of 1s is $|C|\ELL$. Since $C$ is cyclic, the number of codewords in $C$ that have a 1 in position $i\in [|V|]$ is the same as the number of codewords that have a 1 in position $j$, for any $j\in [|V|]$.
Thus, for each one of the $|V|$ positions there are exactly $|C|\ELL/|V|$ codewords in $C$ with a 1 in that position.

Let $V'\sqsubseteq V$ be of size $(63/64)|V|$. Let $J$ be the set of $|V|/64$ positions that correspond to the vectors of $V$ that are not in $V'$.
We will now modify the codewords of $C$ as follows.
For each $j\in J$ and codeword $c\in C$ we
set $c[j]$ to 0.
The total number of 1s across all codewords in $C$ is therefore reduced from $|C|\mu$ by exactly
\[
  \frac{|V|}{64} \cdot
  \frac{|C|\ELL}{|V|} = \frac{|C|\ELL}{64}.
\]
The number of codewords of $C$ that have lost $\ELL/16$ or more 1s is
therefore at most
\[
 \frac{|C|\ELL}{64}\,\Big/\,\frac{\ELL}{16}=\frac{|C|}4.
\]
Let $C'\subseteq C$ be the set of codewords $c$
that have lost less than $\ELL/16$ 1s and for which $\ball{c}$ is good. Since there are at least $|C|/2$
good balls, $|C'|\geq |C|/4$.

Let the code $C''$ be obtained from $C'$ by replacing, for each codeword in $C'$, every removed 1 with a 1 at some other arbitrary position that is not in $J$.
Thus, every $c''\in C''$ has Hamming weight $\ELL$ and belongs to the good ball $\ball{c'}$, where $c''$ was obtained from $c'\in C'$.
Hence $|C''|=|C'|\geq |C|/4$.
Every codeword of $C''$, seen as a sub-multiset of $V$, only contains vectors from the sub-multiset $V'$.
From the definition of a good ball we have that at least $|C|/4$ distinct vector sums can be
obtained by adding $\ELL$ vectors from $V'$.
Thus,
\[
  \Vsum(V')
  \,\geq\,
  \frac14 (\ELL-1)^{(\ELL/9)}
  \,\geq\,
  \ELL^{(\ELL/10)}
\]
when $\ELL>40$. This completes the proof of Lemma~\ref{lem:vecsum}.
%\end{proof}

%\newpage
\section*{Acknowledgements}
RC would like to thank Elad Verbin, Kasper Green Larsen, Qin Zhang and the members of CTIC for helpful and insightful discussions about lower bounds during a visit to Aarhus University. We thank Kasper Green Larsen in particular for pointing out that the cell-probe lower bounds we give are in fact tight. Some of the work on this paper has been carried out during RC's visit at the University of Washington.

\printbibliography

\end{document}